\documentclass[runningheads,orivec]{llncs}
\usepackage{hyperref,times,breakurl,latexsym,color}
\usepackage{amssymb}  
\usepackage{wrapfig}

\DeclareSymbolFont{frenchscript}{OMS}{ztmcm}{m}{n}
\DeclareMathSymbol{\Bi}{\mathord}{frenchscript}{66}  
\DeclareMathSymbol{\Ri}{\mathord}{frenchscript}{82}  
\DeclareMathSymbol{\U}{\mathord}{frenchscript}{85}   
\DeclareMathSymbol{\HC}{\mathord}{frenchscript}{67}  
\DeclareMathSymbol{\Lab}{\mathord}{frenchscript}{76} 
\DeclareMathSymbol{\pow}{\mathord}{frenchscript}{80} 
\newfont{\fsc}{eusm10 scaled 1100}      
\newfont{\fscs}{eusm10 scaled 800}      
\def\processfont#1{\textrm{\fsc #1}}
\def\AA{\processfont{A}}                             
\def\NN{\processfont{N}}                             
\def\NNs{\textrm{\fscs N}}                           
\def\SS{\processfont{S}}                             
\def\TT{\processfont{T}}                             
\def\FF{\processfont{F}}                             
\def\MM{\processfont{M}}                             
\def\PP{\processfont{P}}                             
\def\sk{\mathfrak{s}}                                
\def\tk{\mathfrak{t}}                                

\spnewtheorem{observation}{Observation}{\bfseries}{\itshape}
\newcommand{\weg}[1]{}

\makeatletter                                        
\def\moverlay{\mathpalette\mov@rlay}
\def\mov@rlay#1#2{\leavevmode\vtop{%
   \baselineskip\z@skip \lineskiplimit-\maxdimen
   \ialign{\hfil$\m@th#1##$\hfil\cr#2\crcr}}}
\newcommand{\charfusion}[3][\mathord]{
    #1{\ifx#1\mathop\vphantom{#2}\fi
        \mathpalette\mov@rlay{#2\cr#3}
      }
    \ifx#1\mathop\expandafter\displaylimits\fi}
\makeatother
\newcommand{\dcup}{\charfusion[\mathbin]{\cup}{\mbox{\Large$\cdot$}}}

\newcommand{\rec}[1]{\langle #1 \rangle}             
\newcommand{\plat}[1]{\raisebox{0pt}[0pt][0pt]{#1}}  
\def\precond#1{{\vphantom{#1}}^\bullet #1}             
\def\postcond#1{{#1}^\bullet}                          
\newcommand{\denote}[1]{[\hspace{-1.4pt}[#1]\hspace{-1.4pt}]}  
\newcommand{\A}{C}                                   
\newcommand{\B}{D}                                   
\newcommand{\E}{P}                                   
\newcommand{\F}{Q}                                   
\newcommand{\SC}{{S}}                                
\newcommand{\RS}{{\cal S}}                           
\newcommand{\SX}{S}                                  
\newcommand{\tu}{t}                                  
\newcommand{\T}{{\rm T}}                             
\newcommand{\IN}{
    \ensuremath{%
        \mathop{\rm I\mkern-2.5mu N}%
        \nolimits%
    }%
}
\def\restrictedto{\mathop\upharpoonright}            
\newcommand{\obis}[2]{\mathrel{_{#1}\,		     
	\raisebox{.3ex}{$\underline{\makebox[.7em]{$\leftrightarrow$}}$}
                  \,_{#2}}}
\newcommand{\bis}[1]{\obis{}{#1}}		     
\newcommand{\we}{I}                                  
\newcommand{\Thm}[1]{Theorem~\ref{thm:#1}}

\newcommand{\Prop}[1]{Proposition~\ref{prop:#1}}

\newcommand{\Def}[1]{Definition~\ref{df:#1}}

\newcommand{\Sect}[1]{Section~\ref{sec:#1}}
\newcommand{\SSect}[1]{Section~\ref{ssec:#1}}

\newcommand{\Fig}[1]{Figure~\ref{fig:#1}}
\newcommand{\Tab}[1]{Table~\ref{tab:#1}}

\makeatletter
\def\comesfrom{\@transition\leftarrowfill}
\def\goesto{\@transition\rightarrowfill}
\def\ngoesto{\@transition\nrightarrowfill}
\def\@transition#1{\@@transition{#1}}
\newbox\@transbox
\newbox\@arrowbox
\newbox\@downbox
\def\@@transition#1#2%
   {\setbox\@transbox\hbox
      {\vrule height 1.5ex depth .9ex width 0ex\hskip0.25em$\scriptstyle#2$\hskip0.25em}
   \ifdim\wd\@transbox<1.5em
      \setbox\@transbox\hbox to 1.5em{\hfil\box\@transbox\hfil}\fi
   \setbox\@arrowbox\hbox to \wd\@transbox{#1}
   \ht\@arrowbox\z@\dp\@arrowbox\z@
   \setbox\@transbox\hbox{$\mathop{\box\@arrowbox}\limits^{\box\@transbox}$}
   \dp\@transbox\z@\ht\@transbox 10pt
   \mathrel{\box\@transbox}}
\def\nrightarrowfill{$\m@th\mathord-\mkern-6mu%
  \cleaders\hbox{$\mkern-2mu\mathord-\mkern-2mu$}\hfill
  \mkern-6mu\mathord\not\mkern-2mu\mathord\rightarrow$}
\makeatother 
\newcommand{\ar}[1]{\mathrel{\goesto{#1}}}           
\begin{document}

\titlerunning{Structure Preserving Bisimilarity}
\title{Structure Preserving Bisimilarity,\texorpdfstring{\\}{}
Supporting an Operational Petri Net Semantics of CCSP}
\authorrunning{R.J.\ van Glabbeek}
\author{Rob J. van Glabbeek \inst{1}\inst{2}
}
\institute{
NICTA\thanks{NICTA is funded by the Australian Government through the
    Department of Communications and the Australian Research Council
    through the ICT Centre of Excellence Program.}, Sydney, Australia\and
Computer Science and Engineering, UNSW, Sydney, Australia
}
\maketitle

\setcounter{footnote}{0}
\begin{abstract}
In 1987 Ernst-R\"udiger Olderog provided an operational Petri net
semantics for a subset of CCSP, the union of Milner's CCS and Hoare's CSP\@.
It assigns to each process term in the subset a labelled, safe place/transition net.
To demonstrate the correctness of the approach, Olderog established agreement\linebreak[3]
(1) with the standard interleaving semantics of CCSP up to strong bisimulation equivalence, and
(2) with standard denotational interpretations of CCSP operators in terms of Petri nets
up to a suitable semantic equivalence that fully respects the causal structure of nets.
For the latter he employed a linear-time semantic equivalence, namely having the same causal nets.

\hspace{2em}%
This paper strengthens (2), employing a novel branching-time version of this
semantics---\emph{structure preserving bisimilarity}---that moreover preserves inevitability.
I establish that it is a congruence for the operators of CCSP\@.
\end{abstract}

\section{Introduction}

The system description languages CCS and CSP have converged to one
theory of processes which---following a suggestion of
M.\  Nielsen---was called ``CCSP'' in \cite{Old87}.
The standard semantics of this language is in
terms of labelled transition systems modulo strong bisimilarity, or
some coarser semantic equivalence. In the case of CCS, a labelled
transition system is obtained by taking as states the closed CCS
expressions, and as transitions those that are derivable from a
collection of rules by induction on the structure of these expressions \cite{Mi90ccs};
this is called a \emph{(structural) operational semantics}~\cite{Pl04}.
The semantics of CSP was originally given in quite a different way
\cite{BHR84,Ho85}, but \cite{OH86} provided an operational semantics
of CSP in the same style as the one of CCS, and showed its consistency
with the original semantics.

Such semantics abstract from concurrency relations between actions by reducing concurrency
to interleaving. An alternative semantics, explicitly modelling concurrency relations,
requires models like Petri nets \cite{Re85} or event structures \cite{NPW81,Wi87a}.
In \cite{Wi87a,LG91} non-interleaving semantics for variants of CCSP are
given in terms of event structures. However, infinite event structures
are needed to model simple systems involving loops, whereas Petri
nets, like labelled transition systems, offer finite representations
for some such systems. Denotational semantics in terms of Petri nets of the
essential CCSP operators are given in \cite{GM84,Wi84,GV87}---see
\cite{Old91} for more references. Yet a satisfactory denotational
Petri net semantics treating recursion has to my knowledge not been proposed. 

Olderog \cite{Old87,Old91} closed this gap by giving an
operational net semantics in the style of \cite{Pl04,Mi90ccs} for a
subset of CCSP including recursion---to be precise: \emph{guarded} recursion.
To demonstrate the correctness of his approach, Olderog proposed two
fundamental properties such a semantics should have, and established
that both of them hold \cite{Old91}:\vspace{-1ex}
\begin{itemize}
\item \emph{Retrievability}. The standard interleaving semantics for
  process terms should be retrievable from the net semantics.
\item \emph{Concurrency}. The net semantics should represent the intended
  concurrency of process terms.
\vspace{-1ex}
\end{itemize}
The second requirement was not met by an earlier operational net
semantics from \cite{DDM87}.

To formalise the first requirement, Olderog notes that a Petri net induces
a labelled transition system through the firing relation between
markings---the \emph{interleaving case graph}---and requires that the
interpretation of any CCSP expression as a state in a labelled
transition system through the standard interleaving semantics of CCSP
should be strongly bisimilar to the interpretation of this expression
as a marking in the interleaving case graph induced by its net semantics.

To formalise the second requirement, he notes that the intended
concurrency of process terms is clearly represented in the standard
denotational semantics of CCSP operators \cite{GM84,Wi84,GV87}, and
thus requires that the result of applying a CCSP operator to its
arguments according to this denotational semantics yields a similar
result as doing this according to the new operational semantics.
The correct representation of recursion follows from the correct
representation of the other operators through the observation that a
recursive call has the very same interpretation as a Petri net as its unfolding.

A crucial parameter in this formalisation is the meaning of ``similar''.
A logical choice would be semantic equivalence according to one of the
non-interleaving equivalences found in the literature, where a finer
or more discriminating semantics gives a stronger result.
To match the concurrency requirement, this equivalence should
\emph{respect concurrency}, in that it only identifies nets which
display the same concurrency relations.
In this philosophy, the semantics of a CCSP expression is not so much
a Petri net, but a semantic equivalence class of Petri nets, i.e.\ a
Petri net after abstraction from irrelevant differences between nets.
For this idea to be entirely consistent, one needs to require that the
chosen equivalence is a congruence for all CCSP constructs, so that
the meaning of the composition of two systems, both represented as
equivalence classes of nets, is independent of the choice of
representative Petri nets within these classes.

Instead of selecting such an equivalence, Olderog instantiates
``similar'' in the above formalisation of the second requirement
with \emph{strongly bisimilar}, a new relation between nets that
should not be confused with the traditional relation of strong
bisimilarity \mbox{between} labelled transition systems. As shown in
\cite{ABS91}, strong bisimilarity fails to be an equivalence: it is
reflexive and symmetric, but not transitive.

As pointed out in \cite[Page 37]{Old91} this general shortcoming of
strong bisimilarity ``does not affect the purpose of this relation'' in
that book: there it ``serves as an auxiliary notion in proving that
structurally different nets are causally equivalent''.
Here \emph{causal equivalence} means having the same causal nets,
where \emph{causal nets} \cite{Pe77,Re13} model concurrent
computations or executions of Petri nets.
So in effect Olderog does choose a semantic equivalence on Petri nets,
namely having the same concurrent computations as modelled by causal nets.
This equivalence fully respects concurrency.

\subsection{Structure preserving bisimilarity}\label{sec:spectrum}

The contribution of the present paper is a strengthening of this
choice of a semantic equivalence on Petri nets.  I propose the novel
\emph{structure preserving bisimulation} equivalence on Petri nets,
and establish that the result of applying a CCSP operator to its
arguments according to the standard denotational semantics yields a
structure preserving bisimilar result as doing this according to
Olderog's operational semantics.  The latter is an immediate
consequence of the observation that structure preserving bisimilarity
between two nets is implied by Olderog's strong bisimilarity.
\vspace{-2ex}

\begin{figure}[htb]
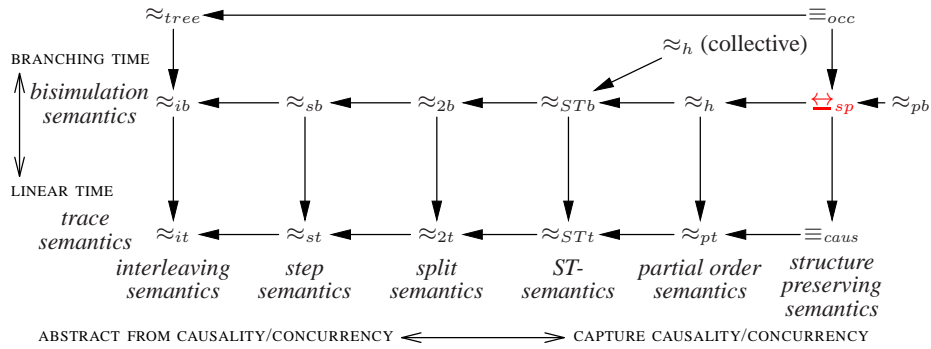

\expandafter\ifx\csname graph\endcsname\relax
   \csname newbox\expandafter\endcsname\csname graph\endcsname
\fi
\ifx\graphtemp\undefined
  \csname newdimen\endcsname\graphtemp
\fi
\expandafter\setbox\csname graph\endcsname
 =\vtop{\vskip 0pt\hbox{%
    \graphtemp=\baselineskip
    \multiply\graphtemp by -1
    \divide\graphtemp by 2
    \advance\graphtemp by .5ex
    \advance\graphtemp by 0.534in
    \rlap{\kern 0.552in\lower\graphtemp\hbox to 0pt{\hss \raisebox{-2pt}[0pt][0pt]{\textit{bisimulation}}\hss}}%
    \graphtemp=\baselineskip
    \multiply\graphtemp by 1
    \divide\graphtemp by 2
    \advance\graphtemp by .5ex
    \advance\graphtemp by 0.534in
    \rlap{\kern 0.552in\lower\graphtemp\hbox to 0pt{\hss \textit{semantics}\hss}}%
    \graphtemp=\baselineskip
    \multiply\graphtemp by -1
    \divide\graphtemp by 2
    \advance\graphtemp by .5ex
    \advance\graphtemp by 1.224in
    \rlap{\kern 0.529in\lower\graphtemp\hbox to 0pt{\hss \textit{trace}\hss}}%
    \graphtemp=\baselineskip
    \multiply\graphtemp by 1
    \divide\graphtemp by 2
    \advance\graphtemp by .5ex
    \advance\graphtemp by 1.224in
    \rlap{\kern 0.529in\lower\graphtemp\hbox to 0pt{\hss \raisebox{2pt}[0pt][0pt]{\textit{semantics}}\hss}}%
    \graphtemp=\baselineskip
    \multiply\graphtemp by -1
    \divide\graphtemp by 2
    \advance\graphtemp by .5ex
    \advance\graphtemp by 1.454in
    \rlap{\kern 0.989in\lower\graphtemp\hbox to 0pt{\hss \raisebox{-2pt}[0pt][0pt]{\textit{interleaving}}\hss}}%
    \graphtemp=\baselineskip
    \multiply\graphtemp by 1
    \divide\graphtemp by 2
    \advance\graphtemp by .5ex
    \advance\graphtemp by 1.454in
    \rlap{\kern 0.989in\lower\graphtemp\hbox to 0pt{\hss \textit{semantics}\hss}}%
    \graphtemp=\baselineskip
    \multiply\graphtemp by -1
    \divide\graphtemp by 2
    \advance\graphtemp by .5ex
    \advance\graphtemp by 1.454in
    \rlap{\kern 1.678in\lower\graphtemp\hbox to 0pt{\hss \raisebox{-2pt}[0pt][0pt]{\textit{step}}\hss}}%
    \graphtemp=\baselineskip
    \multiply\graphtemp by 1
    \divide\graphtemp by 2
    \advance\graphtemp by .5ex
    \advance\graphtemp by 1.454in
    \rlap{\kern 1.678in\lower\graphtemp\hbox to 0pt{\hss \textit{semantics}\hss}}%
    \graphtemp=\baselineskip
    \multiply\graphtemp by -1
    \divide\graphtemp by 2
    \advance\graphtemp by .5ex
    \advance\graphtemp by 1.454in
    \rlap{\kern 2.368in\lower\graphtemp\hbox to 0pt{\hss \raisebox{-2pt}[0pt][0pt]{\textit{split}}\hss}}%
    \graphtemp=\baselineskip
    \multiply\graphtemp by 1
    \divide\graphtemp by 2
    \advance\graphtemp by .5ex
    \advance\graphtemp by 1.454in
    \rlap{\kern 2.368in\lower\graphtemp\hbox to 0pt{\hss \textit{semantics}\hss}}%
    \graphtemp=\baselineskip
    \multiply\graphtemp by -1
    \divide\graphtemp by 2
    \advance\graphtemp by .5ex
    \advance\graphtemp by 1.454in
    \rlap{\kern 3.057in\lower\graphtemp\hbox to 0pt{\hss \raisebox{-2pt}[0pt][0pt]{\textit{ST-}}\hss}}%
    \graphtemp=\baselineskip
    \multiply\graphtemp by 1
    \divide\graphtemp by 2
    \advance\graphtemp by .5ex
    \advance\graphtemp by 1.454in
    \rlap{\kern 3.057in\lower\graphtemp\hbox to 0pt{\hss \textit{semantics}\hss}}%
    \graphtemp=\baselineskip
    \multiply\graphtemp by -1
    \divide\graphtemp by 2
    \advance\graphtemp by .5ex
    \advance\graphtemp by 1.454in
    \rlap{\kern 3.747in\lower\graphtemp\hbox to 0pt{\hss \raisebox{-2pt}[0pt][0pt]{\textit{partial order}}\hss}}%
    \graphtemp=\baselineskip
    \multiply\graphtemp by 1
    \divide\graphtemp by 2
    \advance\graphtemp by .5ex
    \advance\graphtemp by 1.454in
    \rlap{\kern 3.747in\lower\graphtemp\hbox to 0pt{\hss \textit{semantics}\hss}}%
    \graphtemp=\baselineskip
    \multiply\graphtemp by -2
    \divide\graphtemp by 2
    \advance\graphtemp by .5ex
    \advance\graphtemp by 1.500in
    \rlap{\kern 4.437in\lower\graphtemp\hbox to 0pt{\hss \raisebox{-2pt}[0pt][0pt]{\textit{structure}}\hss}}%
    \graphtemp=.5ex
    \advance\graphtemp by 1.500in
    \rlap{\kern 4.437in\lower\graphtemp\hbox to 0pt{\hss \textit{preserving}\hss}}%
    \graphtemp=\baselineskip
    \multiply\graphtemp by 2
    \divide\graphtemp by 2
    \advance\graphtemp by .5ex
    \advance\graphtemp by 1.500in
    \rlap{\kern 4.437in\lower\graphtemp\hbox to 0pt{\hss \raisebox{2pt}[0pt][0pt]{\textit{semantics}}\hss}}%
    \graphtemp=.5ex
    \advance\graphtemp by 1.224in
    \rlap{\kern 0.989in\lower\graphtemp\hbox to 0pt{\hss $\approx_{it}$\hss}}%
    \graphtemp=.5ex
    \advance\graphtemp by 0.534in
    \rlap{\kern 0.989in\lower\graphtemp\hbox to 0pt{\hss $\approx_{ib}$\hss}}%
    \graphtemp=.5ex
    \advance\graphtemp by 0.075in
    \rlap{\kern 0.989in\lower\graphtemp\hbox to 0pt{\hss $\approx_{tree}$\hss}}%
    \graphtemp=.5ex
    \advance\graphtemp by 1.224in
    \rlap{\kern 1.678in\lower\graphtemp\hbox to 0pt{\hss $\approx_{st}$\hss}}%
    \graphtemp=.5ex
    \advance\graphtemp by 0.534in
    \rlap{\kern 1.678in\lower\graphtemp\hbox to 0pt{\hss $\approx_{sb}$\hss}}%
    \graphtemp=.5ex
    \advance\graphtemp by 1.224in
    \rlap{\kern 2.368in\lower\graphtemp\hbox to 0pt{\hss $\approx_{2t}$\hss}}%
    \graphtemp=.5ex
    \advance\graphtemp by 0.534in
    \rlap{\kern 2.368in\lower\graphtemp\hbox to 0pt{\hss $\approx_{2b}$\hss}}%
    \graphtemp=.5ex
    \advance\graphtemp by 1.224in
    \rlap{\kern 3.057in\lower\graphtemp\hbox to 0pt{\hss $\approx_{STt}$\hss}}%
    \graphtemp=.5ex
    \advance\graphtemp by 0.534in
    \rlap{\kern 3.057in\lower\graphtemp\hbox to 0pt{\hss $\approx_{STb}$\hss}}%
    \graphtemp=.5ex
    \advance\graphtemp by 1.224in
    \rlap{\kern 3.747in\lower\graphtemp\hbox to 0pt{\hss $\approx_{pt}$\hss}}%
    \graphtemp=.5ex
    \advance\graphtemp by 0.534in
    \rlap{\kern 3.747in\lower\graphtemp\hbox to 0pt{\hss $\approx_{h}$\hss}}%
    \graphtemp=.5ex
    \advance\graphtemp by 0.236in
    \rlap{\kern 3.931in\lower\graphtemp\hbox to 0pt{\hss $\approx_{h}$ (collective)\hss}}%
    \graphtemp=.5ex
    \advance\graphtemp by 1.224in
    \rlap{\kern 4.437in\lower\graphtemp\hbox to 0pt{\hss $\equiv_{\it caus}$\hss}}%
    \graphtemp=.5ex
    \advance\graphtemp by 0.534in
    \rlap{\kern 4.437in\lower\graphtemp\hbox to 0pt{\hss \color{red}$\bis{sp}$\hss}}%
    \graphtemp=.5ex
    \advance\graphtemp by 0.075in
    \rlap{\kern 4.437in\lower\graphtemp\hbox to 0pt{\hss $\equiv_{occ}$\hss}}%
    \graphtemp=.5ex
    \advance\graphtemp by 0.534in
    \rlap{\kern 4.851in\lower\graphtemp\hbox to 0pt{\hss $\approx_{pb}$\hss}}%
    \special{sh 1.000}%
    \special{pn 1}%
    \special{pa 1014 1049}%
    \special{pa 989 1149}%
    \special{pa 964 1049}%
    \special{pa 1014 1049}%
    \special{fp}%
    \special{pn 8}%
    \special{pa 989 609}%
    \special{pa 989 1049}%
    \special{fp}%
    \special{sh 1.000}%
    \special{pn 1}%
    \special{pa 1703 1049}%
    \special{pa 1678 1149}%
    \special{pa 1653 1049}%
    \special{pa 1703 1049}%
    \special{fp}%
    \special{pn 8}%
    \special{pa 1678 609}%
    \special{pa 1678 1049}%
    \special{fp}%
    \special{sh 1.000}%
    \special{pn 1}%
    \special{pa 2393 1049}%
    \special{pa 2368 1149}%
    \special{pa 2343 1049}%
    \special{pa 2393 1049}%
    \special{fp}%
    \special{pn 8}%
    \special{pa 2368 609}%
    \special{pa 2368 1049}%
    \special{fp}%
    \special{sh 1.000}%
    \special{pn 1}%
    \special{pa 3082 1049}%
    \special{pa 3057 1149}%
    \special{pa 3032 1049}%
    \special{pa 3082 1049}%
    \special{fp}%
    \special{pn 8}%
    \special{pa 3057 609}%
    \special{pa 3057 1049}%
    \special{fp}%
    \special{sh 1.000}%
    \special{pn 1}%
    \special{pa 3772 1049}%
    \special{pa 3747 1149}%
    \special{pa 3722 1049}%
    \special{pa 3772 1049}%
    \special{fp}%
    \special{pn 8}%
    \special{pa 3747 609}%
    \special{pa 3747 1049}%
    \special{fp}%
    \special{sh 1.000}%
    \special{pn 1}%
    \special{pa 4462 1049}%
    \special{pa 4437 1149}%
    \special{pa 4412 1049}%
    \special{pa 4462 1049}%
    \special{fp}%
    \special{pn 8}%
    \special{pa 4437 609}%
    \special{pa 4437 1049}%
    \special{fp}%
    \special{sh 1.000}%
    \special{pn 1}%
    \special{pa 1014 360}%
    \special{pa 989 460}%
    \special{pa 964 360}%
    \special{pa 1014 360}%
    \special{fp}%
    \special{pn 8}%
    \special{pa 989 149}%
    \special{pa 989 360}%
    \special{fp}%
    \special{sh 1.000}%
    \special{pn 1}%
    \special{pa 4462 360}%
    \special{pa 4437 460}%
    \special{pa 4412 360}%
    \special{pa 4462 360}%
    \special{fp}%
    \special{pn 8}%
    \special{pa 4437 149}%
    \special{pa 4437 360}%
    \special{fp}%
    \special{sh 1.000}%
    \special{pn 1}%
    \special{pa 1249 100}%
    \special{pa 1149 75}%
    \special{pa 1249 50}%
    \special{pa 1249 100}%
    \special{fp}%
    \special{pn 8}%
    \special{pa 4299 75}%
    \special{pa 1249 75}%
    \special{fp}%
    \special{sh 1.000}%
    \special{pn 1}%
    \special{pa 1226 1249}%
    \special{pa 1126 1224}%
    \special{pa 1226 1199}%
    \special{pa 1226 1249}%
    \special{fp}%
    \special{pn 8}%
    \special{pa 1540 1224}%
    \special{pa 1226 1224}%
    \special{fp}%
    \special{sh 1.000}%
    \special{pn 1}%
    \special{pa 1916 1249}%
    \special{pa 1816 1224}%
    \special{pa 1916 1199}%
    \special{pa 1916 1249}%
    \special{fp}%
    \special{pn 8}%
    \special{pa 2230 1224}%
    \special{pa 1916 1224}%
    \special{fp}%
    \special{sh 1.000}%
    \special{pn 1}%
    \special{pa 2606 1249}%
    \special{pa 2506 1224}%
    \special{pa 2606 1199}%
    \special{pa 2606 1249}%
    \special{fp}%
    \special{pn 8}%
    \special{pa 2891 1224}%
    \special{pa 2606 1224}%
    \special{fp}%
    \special{sh 1.000}%
    \special{pn 1}%
    \special{pa 3324 1249}%
    \special{pa 3224 1224}%
    \special{pa 3324 1199}%
    \special{pa 3324 1249}%
    \special{fp}%
    \special{pn 8}%
    \special{pa 3609 1224}%
    \special{pa 3324 1224}%
    \special{fp}%
    \special{sh 1.000}%
    \special{pn 1}%
    \special{pa 3985 1249}%
    \special{pa 3885 1224}%
    \special{pa 3985 1199}%
    \special{pa 3985 1249}%
    \special{fp}%
    \special{pn 8}%
    \special{pa 4270 1224}%
    \special{pa 3985 1224}%
    \special{fp}%
    \special{sh 1.000}%
    \special{pn 1}%
    \special{pa 1226 559}%
    \special{pa 1126 534}%
    \special{pa 1226 509}%
    \special{pa 1226 559}%
    \special{fp}%
    \special{pn 8}%
    \special{pa 1540 534}%
    \special{pa 1226 534}%
    \special{fp}%
    \special{sh 1.000}%
    \special{pn 1}%
    \special{pa 1916 559}%
    \special{pa 1816 534}%
    \special{pa 1916 509}%
    \special{pa 1916 559}%
    \special{fp}%
    \special{pn 8}%
    \special{pa 2230 534}%
    \special{pa 1916 534}%
    \special{fp}%
    \special{sh 1.000}%
    \special{pn 1}%
    \special{pa 2606 559}%
    \special{pa 2506 534}%
    \special{pa 2606 509}%
    \special{pa 2606 559}%
    \special{fp}%
    \special{pn 8}%
    \special{pa 2891 534}%
    \special{pa 2606 534}%
    \special{fp}%
    \special{sh 1.000}%
    \special{pn 1}%
    \special{pa 3324 559}%
    \special{pa 3224 534}%
    \special{pa 3324 509}%
    \special{pa 3324 559}%
    \special{fp}%
    \special{pn 8}%
    \special{pa 3609 534}%
    \special{pa 3324 534}%
    \special{fp}%
    \special{sh 1.000}%
    \special{pn 1}%
    \special{pa 3293 449}%
    \special{pa 3192 470}%
    \special{pa 3271 404}%
    \special{pa 3293 449}%
    \special{fp}%
    \special{pn 8}%
    \special{pa 3544 300}%
    \special{pa 3282 426}%
    \special{fp}%
    \special{sh 1.000}%
    \special{pn 1}%
    \special{pa 3985 559}%
    \special{pa 3885 534}%
    \special{pa 3985 509}%
    \special{pa 3985 559}%
    \special{fp}%
    \special{pn 8}%
    \special{pa 4299 534}%
    \special{pa 3985 534}%
    \special{fp}%
    \special{sh 1.000}%
    \special{pn 1}%
    \special{pa 4675 559}%
    \special{pa 4575 534}%
    \special{pa 4675 509}%
    \special{pa 4675 559}%
    \special{fp}%
    \special{pn 8}%
    \special{pa 4713 534}%
    \special{pa 4675 534}%
    \special{fp}%
    \graphtemp=.5ex
    \advance\graphtemp by 0.305in
    \rlap{\kern 0.138in\lower\graphtemp\hbox to 0pt{\hss \makebox[0pt][l]{\sc\scriptsize branching time}\hss}}%
    \graphtemp=.5ex
    \advance\graphtemp by 0.994in
    \rlap{\kern 0.138in\lower\graphtemp\hbox to 0pt{\hss \makebox[0pt][l]{\sc\scriptsize linear time}\hss}}%
    \special{pa 159 480}%
    \special{pa 184 380}%
    \special{fp}%
    \special{pa 209 480}%
    \special{pa 184 380}%
    \special{fp}%
    \special{pa 209 819}%
    \special{pa 184 919}%
    \special{fp}%
    \special{pa 159 819}%
    \special{pa 184 919}%
    \special{fp}%
    \special{pa 184 380}%
    \special{pa 184 919}%
    \special{fp}%
    \graphtemp=.5ex
    \advance\graphtemp by 1.753in
    \rlap{\kern 1.218in\lower\graphtemp\hbox to 0pt{\hss {\sc\scriptsize abstract from causality/concurrency}\hss}}%
    \graphtemp=.5ex
    \advance\graphtemp by 1.753in
    \rlap{\kern 3.862in\lower\graphtemp\hbox to 0pt{\hss {\sc\scriptsize capture causality/concurrency}\hss}}%
    \special{pa 2285 1789}%
    \special{pa 2185 1764}%
    \special{fp}%
    \special{pa 2285 1739}%
    \special{pa 2185 1764}%
    \special{fp}%
    \special{pa 2934 1739}%
    \special{pa 3034 1764}%
    \special{fp}%
    \special{pa 2934 1789}%
    \special{pa 3034 1764}%
    \special{fp}%
    \special{pa 2185 1764}%
    \special{pa 3034 1764}%
    \special{fp}%
    \hbox{\vrule depth1.828in width0pt height 0pt}%
    \kern 5.057in
  }%
}%
\centerline{\box\graph}
\caption{A spectrum of semantic equivalences on Petri nets}\label{fig:spectrum}
\vspace{-2ex}
\end{figure}

\Fig{spectrum} shows a map of some equivalence relations on nets found
in the literature, in relation to the new structure preserving
bisimilarity, $\bis{sp}$. The equivalences become finer when moving up or to
the right; thus coarser or less discriminating when following the arrows.
The rectangle from $\approx_{it}$ to $\approx_h$ is taken
from \cite{vG90}. The vertical axis is the \emph{linear time -- branching
time spectrum}, with \emph{trace equivalence} at the bottom and
\emph{(strong) bisimulation equivalence}, or \emph{bisimilarity}, at the top.
A host of intermediate equivalences is discussed in \cite{vG01}.
The key difference is that \emph{linear time} equivalences, like trace
equivalence, only consider the set of possible executions of a process,
whereas \emph{branching time} equivalences, like bisimilarity,
additionally take into account at which point the choice between two
executions is made.
The horizontal axis indicates to what extent concurrency information is
taken into account. \emph{Interleaving} equivalences---on the
left---fully abstract from concurrency by reducing it to arbitrary
interleaving; \emph{step} equivalences additionally take into account
the possibility that two concurrent actions happen at exactly the same
moment; \emph{split} equivalences recognise the beginning and end of
actions, which here are regarded to be durational, thereby capturing
some information about their overlap in time; \emph{ST-} or
\emph{interval} equivalences fully capture concurrency information as
far as possible by considering durational actions overlapping in time;
and \emph{partial order} equivalences capture the causal links
between actions, and thereby all concurrency.
By taking the product of these two axes, one obtains a two-dimensional
spectrum of equivalence relations, with entries like
\emph{interleaving bisimulation} equivalence $\approx_{ib}$ and
\emph{partial order trace} equivalence $\approx_{pt}$. For the right upper corner
several partial order bisimulation equivalences were proposed in the
literature; according to~\cite{GG01} the \emph{history preserving bisimulation}
equivalence $\approx_h$, originally proposed by \cite{RT88}, is the
coarsest one that fully captures the interplay between causality and branching time.

The causal equivalence employed by Olderog, $\equiv_{\it caus}$, is a
linear time equivalence strictly finer than $\approx_{pt}$. Since it
preserves information about the number of preplaces of a transition,
it is specific to a model of concurrency based on Petri nets;
i.e.\ there is no obvious counterpart in terms of event structures.
I found only two equivalences in the literature that are finer than
both $\equiv_{\it caus}$ and $\approx_h$, namely \emph{occurrence net equivalence}
\cite{GV87}---$\equiv_{occ}$---and the \emph{place bisimilarity} 
$\approx_{pb}$ of \cite{ABS91}. Two nets are occurrence net
equivalent iff they have isomorphic unfoldings. The \emph{unfolding},
defined in \cite{NPW81}, associates with a given safe Petri net $N$ a
loop-free net---an \emph{occurrence net}---that combines all causal
nets of $N$, together with their branching structure.
This unfolding is similar to the unfolding of a labelled transition
system into a tree, and thus the interleaving counterpart of
occurrence net equivalence is \emph{tree equivalence} \cite{vG01},
identifying two transition systems iff their unfoldings are isomorphic.
The place bisimilarity was inspired by Olderog's strong
bisimilarity, but adapted to make it transitive, and thus an
equivalence relation. My new equivalence $\bis{sp}$ will be shown to
be strictly coarser than $\equiv_{occ}$ and $\approx_{pb}$, yet finer
than both $\equiv_{\it caus}$ and $\approx_h$.

The equivalences discussed above (without the diagonal line in \Fig{spectrum})
are all defined on safe Petri nets. Additionally, the definitions generalise
to unsafe Petri nets. However, there are two possible
interpretations of unsafe Petri nets, called the \emph{collective token}
and the \emph{individual token} interpretation \cite{vG05c}, and this
leads to two versions of history preserving bisimilarity. The history
preserving bisimilarity based on the individual token interpretation
was first defined for Petri nets in \cite{BDKP91}, under the name
\emph{fully concurrent bisimulation} equivalence. At the level of
ST-semantics the collective and individual token interpretations collapse.
The unfolding of unsafe Petri nets, and thereby occurrence net equivalence,
has been defined for the individual token interpretation only
\cite{En91,MMS97,vG05c}, and likewise causal equivalence can be easily
generalised within the individual token interpretation. The new
structure preserving bisimilarity falls in the individual token camp
as well.

\subsection{Criteria for choosing this semantic equivalence}\label{ssec:criteria}

In selecting a new semantic equivalence for reestablishing Olderog's
agreement of operational and denotational interpretations of CCSP operators,
I consider the following requirements on such a semantic equivalence
(with subsequent justifications):
\vspace{-1ex}
\begin{enumerate}
\item it should be a branching time equivalence,
\item it should fully capture causality relations and concurrency
  (and the interplay between causality and branching time),
\item it should respect \emph{inevitability} \cite{MOP89}, meaning that if two
  systems are equivalent, and in one the occurrence of a certain action
  is inevitable, then so is it in the other,
\item it should be \emph{real-time consistent} \cite{GV87}, meaning
  that for every association of execution times to actions,
  assuming that actions happen as soon as they can, the running times
  associated with computations in equivalent systems should be the same,
\item it should be \emph{preserved under action refinement} \cite{CDP87,GG01},
  meaning that if in two equivalent Petri nets the same substitutions of
  nets for actions are made, the resulting nets should again be equivalent,
\item it should be finer than Olderog's causal equivalence,
\item it should not distinguish systems whose behaviours are patently
  the same, such as Petri nets that differ only in their unreachable parts,
\item it should be a congruence for the constructs of CCSP,
\item and it should allow to establish agreement between the
  operational and denotational interpretations of CCSP operators.
\vspace{-1ex}
\end{enumerate}
Requirement 1 is the driving force behind this contribution.
It is motivated by the insight that branching time equivalences
better capture phenomena like deadlock behaviour. 
Since in general a stronger result on the agreement between
operational and denotational semantics is obtained when employing a finer
semantics, I aim for a semantics that fully captures branching time
information, and thus is at least as discriminating as interleaving bisimilarity.

Requirement 2 is an obvious choice when the goal of the project is to
capture concurrency explicitly. The combination of Requirements 1 and 2
then naturally asks for an equivalence that is at least as fine as $\approx_h$.
One might wonder, however, for what reason one bothers to define a
semantics that captures concurrency information. In the
literature, various practical reasons have been given for preferring a
semantics that (partly) respects concurrency and causality over
an interleaving semantics. Three of the more prominent of these reasons
are formulated as requirements 3, 4 and 5 above.

Requirement 3 is manifestly useful when considering liveness properties of systems.
Requirement 4 obviously has some merit when timing is an issue.
Requirement 5 is useful in system design based on stepwise refinement \cite{GG01}.

Requirement 6 is only there so that I can truthfully state to have
strengthened Olderog's agreement between the denotational and
operational semantics, which was stated in terms of causal equivalence.
This requirement will not be needed in my justification for
introducing a new semantic equivalence---and neither will Requirement~2.

Requirement 7 is hardly in need of justification.
The paper \cite{ABS91} lists as a desirable property of semantic
equivalences---one that is not met by their own proposal $\approx_{pb}$---%
that they should not distinguish nets that have isomorphic unfoldings, given
that unfolding a net should not be regarded as changing it behaviour.
When working within the individual token interpretation of nets
I will take this as a suitable formalisation of Requirement~7.

The argument for Requirement 8 has been given earlier in this introduction,
and Requirement 9 underlies my main motivation for selecting a
semantic equivalence in the first place.

\subsection{Applying the criteria}

\begin{table}
\caption{Which requirements are satisfied by the various semantic equivalences\label{requirements}}
\newcommand{\no}{$\times$}
\newcommand{\yes}{\checkmark}
\begin{tabular}{|l|ccc|cc|cc|cc|cc|ccc|c|}
\hline
\hfill \textit{Equivalence}
&&&$\!\!\!\approx_{tree}$ &&&&&&&&&&&$\!\!\equiv_{occ}$&\\
&
&
$\approx_{ib}$ &
&
&
$\approx_{sb}$ &
&
$\approx_{2b}$ &
&
$\!\!\!\!\approx_{STb}$ &
&
$\approx_{h}$ &
&
{\color{red}$\!\!\bis{sp}\!\!$} &
&
$\approx_{pb}$ \\
\textit{Requirement} &
$\approx_{it}$ &
&
&
$\approx_{st}$ &
&
$\approx_{2t}$ &
&
$\approx_{STt}\!\!\!\!$ &
&
$\approx_{pt}$ &
&
$\equiv_{\it caus}\!\!\!\!$ &
&
&
\\
\hline
1. Branching time & \no & \yes& \yes& \no & \yes& \no & \yes& \no & \yes& \no & \yes& \no & \yes& \yes& \yes
\\
2. Causality      & \no & \no & \no & \no & \no & \no & \no & \no & \no & \yes& \yes& \yes& \yes& \yes& \yes
\\
3. Inevitability  & \no & \no & \no & \no & \no & \no & \no & \no & \no & \no & \no & \no & \yes& \yes& \yes
\\
4. Real-time consistency &\no&\no&\no&\no & \no & \no & \no & \no & \yes& \no & \yes& \no & \yes& \yes& \yes
\\
5. Action refinement &\no&\no & \no & \no & \no & \no & \no & \yes& \yes& \yes& \yes& \yes?$\!\!$&\yes?$\!\!$&\yes?$\!\!$&
\\
6. Finer than $\equiv_{\it caus}$ &\no&\no&\no&\no& \no & \no & \no & \no& \no& \no& \no& \yes&\yes&\yes&\yes
\\
7. Coarser than $\equiv_{occ}$ &\yes&\yes&\yes&\yes&\yes&\yes&\yes&\yes&\yes&\yes&\yes&\yes&\yes&\yes&\no
\\
8. Congruence     & \yes & \yes &&&&&&&&&&& \yes &&
\\
9. Operat. $\equiv$ denotat. &\yes&\yes&\no&\yes&\yes&\yes&\yes&\yes&\yes&\yes&\yes&\yes&\yes&\no&
\\
\hline
\end{tabular}
\end{table}

Table~\ref{requirements} tells which of these requirements are satisfied
by the semantic equivalences from Section~\ref{sec:spectrum} (not
considering the one collective token equivalence there).
The first two rows, reporting which equivalences satisfy Requirements 1 and 2,
are well-known; these results follow directly from the definitions.
The third row, reporting on respect for inevitability, is a
contribution of this paper, and will be discussed in \SSect{inevitability},
and delivered in Sections~\ref{sec:inevitability}--\ref{sec:reactive}.

Regarding Row 4, In \cite{GV87} it is established that ST-bisimilarity is real-time
consistent. Moreover, the formal definition is such that if a semantic
equivalence $\approx$ is real-time consistent, then so is any
equivalence finer than $\approx$. Linear time equivalences are not
real-time consistent, and neither is $\approx_{2b}$ \cite{GV97}.

In \cite{GG01} it is established that $\approx_{pt}$ and
$\approx_h$ are preserved under action refinement, but interleaving and
step equivalences are not, because they do not capture enough
information about concurrency. In \cite{vG90} it is shown that
$\approx_{STt}$ and $\approx_{STb}$ are already preserved under action
refinement, whereas by \cite{GV97} split semantics are not. I conjecture that
$\equiv_{\it caus}$ and $\equiv_{occ}$ are also preserved under action
refinement, but I have not seen a formal proof. I also conjecture that
the new $\bis{sp}$ is preserved under action refinement.

Rows 6 and 7 follow as soon as I have formally established the
implications of \Fig{spectrum} (in \Sect{relating}).
As for Row 8, I will show in \Sect{congruence} that $\bis{sp}$ is a
congruence for the operators of CCSP\@. That also $\approx_{it}$ and
$\approx_{ib}$ are congruences for CCSP is well known.
The positive results in Row 9 follow from the fact that Olderog's
strong bisimilarity implies $\bis{sp}$, which will be established in \Sect{strong}.

Requirements 1 and 6 together limit the search space for suitable
equivalence relations to $\equiv_{occ}$, $\approx_{pb}$ and the new $\bis{sp}$.
When dropping Requirement 6, but keeping 2, also $\approx_h$ becomes
in scope. When also dropping 2, but keeping 4, I gain $\approx_{STb}$
as a candidate equivalence. However, both $\approx_h$ and
$\approx_{STb}$ will fall pray to Requirement 3, so also without
Requirements 2 and 6 the search space will be limited to
$\equiv_{occ}$, $\approx_{pb}$ and the new $\bis{sp}$.

Requirement 7 rules out $\approx_{pb}$, as that equivalence makes
distinctions based on unreachable parts of nets \cite{ABS91}.
The indispensable Requirement 9 rules out $\equiv_{occ}$, since that
equivalence distinguishes the operational and denotational semantics of
the CCSP expression $a0+a0$. According to the operational semantics
this expression has only one transition, whereas by the denotational
semantics it has two, and $\equiv_{occ}$ does not collapse identical
choices. The same issue plays in interleaving semantics,
where the operational and denotational transition system semantics of
CCSP do not agree up to tree equivalence. This is one of the main
reasons that bisimilarity is often regarded as the top of the linear
time -- branching time spectrum.

This constitutes the justification for the new equivalence $\bis{sp}$.
\vfill

\subsection{Inevitability}\label{ssec:inevitability}

The meaning of Requirement 3 depends on which type of progress or fairness property
one assumes to guarantee that actions that are due to occur will
actually happen. Lots of fairness assumption are mentioned in the
literature, but, as far as I can tell, they can be
classified in exactly 4 groups: \emph{progress}, \emph{justness},
\emph{weak fairness} and \emph{strong fairness} \cite{GH15a}.
These four groups form a hierarchy, in the sense that one cannot
consistently assume strong fairness while objecting to weak fairness,
or justness while objecting to progress.

Strong and weak fairness deal
with choices that are offered infinitely often. Suppose you have a
shop with only two customers $A$ and $B$ that may return to the shop
to buy something else right after they are served. Then it is unfair
to only serve customer $A$ again and again, while $B$ is continuously
waiting to be served. In case $B$ is not continuously ready to be
served, but sometimes goes home to sleep, yet always returns to wait
for his turn, it is weakly fair to always ignore customer $B$ in
favour of $A$, but not strongly fair.

Weak and strong fairness
assumptions can be made \emph{locally}, pertaining to \emph{some} repeating
choices of the modelled system but not to others, or \emph{globally},
pertaining to all choices of a given type.
Since the real world is largely unfair, strong and weak fairness
assumptions need to be made with great caution, and they will not
appear in this paper.

Justness and progress assumptions, on the other hand, come only in the
global variant, and can be safely assumed much more often.
A progress assumption says that if a system can do some action (that
is not contingent on external input) it will do an action.
In the example of the shop, if there is a customer continuously ready
to be served, and the clerk stands pathetically behind the counter
staring at the customer but not serving anyone, there is a failure of progress.
Without assuming progress, no action is inevitable, because it is
always possible that a system will remain in its initial state without
ever doing anything. Hence the concept of inevitability only makes
sense when assuming at least progress.

Justness \cite{FGHMPT13,GH15a}
says roughly that if a parallel component can make progress (not
contingent on input from outside of this component) it will do so.
Suppose the shop has two counters, each manned by a clerk, and,
whereas customer $A$ is repeatedly served at counter 1,
customer $B$ is ready to be served by counter 2, but is only stared at
by a pathetic clerk. This is not a failure of progress, as in any
state of the system someone will be served eventually. Yet it counts as a failure of justness.
In the context of Petri nets, a failure of justness can easily be
formalised as an execution, during which, from some point onwards, all preplaces
of a given transition remain marked, yet the transition never fires \cite{GH15b}.
One could argue that, when taking concurrency seriously, justness
should be assumed whenever one assumes progress.

Inevitability can be easily expressed in temporal logics like LTL
\cite{Pn77} or CTL \cite{EC82}, and it is well known that strongly
bisimilar transition systems satisfy the same temporal formulas.
This suggests that interleaving bisimilarity already respects
inevitability. However, this conclusion is warranted only when assuming
progress but not justness, or perhaps also when assuming some form of
weak or strong fairness.
The system $C := \rec{X | X\mathbin=aX+bX}$---using the CCSP syntax of
\Sect{CCSP}---repeatedly choosing between the actions $a$ and $b$, is interleaving
bisimilar to the system $D := \rec{Y|Y\mathbin=aY} \| \rec{Z|Z\mathbin=bZ}$, which in
parallel performs infinitely many $a$s and infinitely many $b$s.
Yet, when assuming justness but not weak fairness, the execution of
the action $b$ is inevitable in $D$, but not in $C$.
This shows that when assuming justness but not weak fairness,
interleaving bisimilarity does not respect inevitability. 
The paper \cite{MOP89}, which doesn't use Petri nets as system model,
leaves the precise formulation of a justness assumption for future
work---this task is undertaken in the different context of CCS in \cite{GH15a}.
Also, respect of inevitability as a criterion for judging semantic
equivalences does not occur in \cite{MOP89}, even though ``the partial
order approach'' is shown to be beneficial.

In this paper, assuming justness but not strong or weak fairness, I show that neither $\approx_h$ nor
$\equiv_{\it caus}$ respects inevitability (using infinite nets in my counterexample).
Hence, respecting concurrency appears not quite enough to respect inevitability.
Respect for inevitability, like real-time consistency, is a property
that holds for any equivalence relation finer than one for which it is known to hold already.
So also none of the ST- or interleaving equivalences respects inevitability.
I show that the new equivalence $\bis{sp}$ respects inevitability.
This makes it the coarsest equivalence of \Fig{spectrum} that does so.

\section{CCSP}\label{sec:CCSP}

CCSP is parametrised by the choice of an infinite set $Act$ of
actions, that {\we} will assume to be fixed for this paper.  Just like
the version of CSP from Hoare \cite{Ho85}, the version of CCSP used
here is a typed language, in the sense that with every CCSP process
$P$ an explicit alphabet $\alpha(P) \subseteq Act$ is associated,
which is a superset of the set of all actions the process could
possibly perform. This alphabet is exploited in the definition of the
parallel composition $P \| Q$: actions in the intersection of the
alphabets of $P$ and $Q$ are required to synchronise, whereas all
other actions of $P$ and $Q$ happen independently. Because of this,
processes with different alphabets may never be identified, even if
they can perform the same set of actions and are alike in all other
aspects.  It is for this reason that {\we} interpret CCSP in terms of
{\em typed} Petri nets, with an alphabet as extra component.

I also assume an infinite set $V$ of {\em variable names}. A {\em variable}
is a pair $X_A$ with $X \mathbin\in V$ and $A \subseteq Act$.  The syntax of
(my subset of) CCSP is given by
$$P ::= 0_A ~\mbox{\Large $\,\mid\,$}~ aP ~\mbox{\Large $\,\mid\,$}~ P+P
~\mbox{\Large $\,\mid\,$}~ P \| P ~\mbox{\Large $\,\mid\,$}~ R(P) \mbox{\Large
~$\,\mid\,$}~ X_A ~\mbox{\Large $\,\mid\,$}~ \rec{X_A|\RS}\mbox{ (with }X_A \mathbin\in V_\RS)$$
with $A \mathbin\subseteq Act$, $a \mathbin\in Act$, $R \mathbin\subseteq Act \times Act$,
$X \mathbin\in V$ and $\RS$ a {\em recursive specification}: a set of equations
$\{Y_B = \RS_{Y_B} \mid Y_B \mathbin\in V_\RS\}$ with $V_\RS \subseteq V \times Act$
(the {\em bound variables} of $\RS$) and $\RS_{Y_B}$ a CCSP expression satisfying
$\alpha(\RS_{Y_B})=B$ for all $Y_B \mathbin\in V_\RS$ (were $\alpha(\RS_{Y_B})$ is defined below).
The constant $0_A$ represents a process that is unable to perform any
action. The process $a P$ first performs the action $a$ and then
proceeds as $P$. The process $P+Q$ will behave as either $P$ or $Q$,
$ \| $ is a partially synchronous
parallel composition operator, $R$ a renaming, and
$\rec{X_A|\RS}$ represents the $X_A$-component of a solution of the
system of recursive equations $\RS$.
A CCSP expression $P$ is {\em closed} if every occurrence of a
variable $X_A$ occurs in a subexpression $\rec{Y_B|\RS}$ of $P$ with
$X_A \mathbin\in V_\RS$.

The constant 0 and the variables are indexed with an alphabet.
The alphabet of an arbitrary CCSP expression is given by:
\vspace{-1ex}
\begin{itemize}
\item $\alpha(0_A) = \alpha(X_A) = \alpha(\rec{X_A|\RS}) = A$
\item $\alpha(aP)  = \{a\} \cup \alpha(P)$
\item $\alpha(P+Q) = \alpha(P \|Q) = \alpha(P) \cup \alpha(Q)$
\item $\alpha(R(P))= \{b \mid \exists a \in \alpha(P): (a,b)\in R\}$.
\vspace{-1ex}
\end{itemize}
Substitutions of expressions for variables are allowed only if the
alphabets match. For this reason a recursive specification $\RS$ is
declared syntactically incorrect if \mbox{$\alpha(\RS_{Y_B}) \mathbin{\neq} B$} for some
$Y_B \mathbin\in V_\RS$.
The interleaving semantics of CCSP is given by the labelled transition relation
$\mathord\rightarrow \subseteq \T_{\rm CCSP}\times Act \times\T_{\rm CCSP}$
on the set $\T_{\rm CCSP}$ of closed CCSP terms, where the transitions 
{$\E\ar{a}\F$} (on arbitrary CCSP expressions) are derived from the rules of \Tab{sos CCSP}.
Here $\rec{\E|\RS}$ for $\E$ an expression and $\RS$ a recursive specification
denotes the expression $\E$ in which $\rec{Y_B|\RS_{Y_B}}$ has been substituted for the
variable $Y_B$ for all $Y_B \mathbin\in V_\RS$.

\begin{table}[t]
\caption{Structural operational interleaving semantics of CCSP}
\label{tab:sos CCSP}
\vspace{-2ex}
\normalsize
\begin{center}
\framebox{$\begin{array}{@{}c@{\qquad}c@{\;\quad}c@{}}
a\E \ar{a} \E &
\displaystyle\frac{\E\ar{a} \E'}{\E\|\F \ar{a} \E'\|\F}~~(a\mathbin{\notin}\alpha(\F)) &
\displaystyle\frac{\E \ar{a} \E'}{R(\E) \ar{b} R(\E')}~~((a,b)\mathbin\in R) \\[4ex]

\displaystyle\frac{\E\ar{a} \E'}{\E+\F \ar{a} \E'} & 
\multicolumn{2}{l}{
\displaystyle\frac{\E\ar{a} \E' ,~ \F \ar{a} \F'}{\E\|\F \ar{a} \E'\| \F'}~~
                                 (a\in\alpha(\E)\cap\alpha(\F)) } 
\\[4ex]
\displaystyle\frac{\F \ar{a} \F'}{\E+\F \ar{a} \F'} &
\displaystyle\frac{\F \ar{a} \F'}{\E\|\F \ar{a} \E\|\F'}~~(a\notin\alpha(\E)) &
\displaystyle\frac{\rec{\RS_{X_A}|\RS} \ar{a} \E'}{\rec{X_A|\RS}\ar{a}\E'}
\end{array}$}
\end{center}
\vspace{-3ex}
\end{table}

A CCSP expression is \emph{well-typed} if for any subexpression of the
form $aP$ one has $a\mathbin\in\alpha(P)$ and for any subexpression of the
form $P+Q$ one has $\alpha(P)=\alpha(Q)$. Thus $a0_{\{a\}}+bX_\emptyset$
is not well-typed, although the equivalent expression $a0_{\{a,b\}}+bX_{\{a,b\}}$ is.\linebreak[2]
A recursive specification $\rec{X_A|\RS}$ is \emph{guarded} if each occurrence of a variable
$Y_B\in V_\RS$ in a term $\RS_{Z_C}$ for some $Z_C\mathord\in V_\RS$ lays
within a subterm of $\RS_{Z_C}$ of the form $aP\!$.
Following \cite{Old91}\linebreak[3] I henceforth only consider well-typed
CCSP expressions with guarded recursion.

In Olderog's subset of CCSP, each recursive specification
has only one equation, and renamings must be functions instead of relations.
Here I allow mutual recursion and relational renaming,
where an action may be renamed into a choice of several actions---or possibly none.
This generalisation does not affect any of the proofs in \cite{Old91}.

\newcommand{\Cust}{\textsc{Cus}}
\newcommand{\Clerk}{\textsc{Clk}}
\newcommand{\CustS}{{\textsc{\scriptsize Cus}}}
\newcommand{\ClerkS}{{\textsc{\scriptsize Clk}}}
\newcommand{\enter}{\textit{enter}}
\newcommand{\buy}{\textit{buy}}
\newcommand{\leave}{\textit{leave}}
\newcommand{\serve}{\textit{serve}}
\newcommand{\serves}{\textit{serves}}
\newcommand{\sep}{~}
\newcommand{\ik}{\textrm{I}\,\textit{serves}\,A}
\newcommand{\il}{\textrm{II}\,\textit{serves}\,A}
\newcommand{\jk}{\textrm{I}\,\textit{serves}\,B}
\newcommand{\jl}{\textrm{II}\,\textit{serves}\,B}
\newcommand{\ileaves}{A\,\textit{leaves}}
\newcommand{\jleaves}{B\,\textit{leaves}}
\newcommand{\ienters}{A\,\textit{enters}}
\newcommand{\jenters}{B\,\textit{enters}}
\newcommand{\CU}{\textsc{Cu}}
\newcommand{\CL}{\textsc{Cl}}
\begin{example}
The behaviour of the customer from \SSect{inevitability} could be given by
the recursive specification $\RS_\CustS$:
$$\Cust_{\it Cu} = \enter \sep \buy \sep \leave \sep \Cust_{\it Cu}$$
indicating that the customer keeps coming back to the shop to buy more things.
Here $\enter,\buy,\leave \mathbin\in Act$ and $\Cust\mathbin\in V$.
The customer's alphabet ${\it Cu}$ is $\{\enter,\buy,\leave\}$.
Likewise, the behaviour of the store clerk could be given by the specification $\RS_\ClerkS$:
$$\Clerk_{\it Cl} = \serve \sep \Clerk_{\it Cl}$$
where ${\it Cl}\mathbin=\{\serve\}$. The CCSP processes representing the customer and
the clerk, with their reachable states and labelled transitions
between them, are displayed in \Fig{cc}.
\begin{figure}[htb]
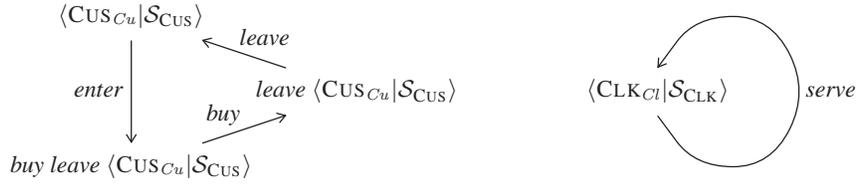

\expandafter\ifx\csname graph\endcsname\relax
   \csname newbox\expandafter\endcsname\csname graph\endcsname
\fi
\ifx\graphtemp\undefined
  \csname newdimen\endcsname\graphtemp
\fi
\expandafter\setbox\csname graph\endcsname
 =\vtop{\vskip 0pt\hbox{%
    \graphtemp=.5ex
    \advance\graphtemp by 0.128in
    \rlap{\kern 0.236in\lower\graphtemp\hbox to 0pt{\hss $\rec{\Cust_{\it Cu} | \RS_\CustS}$\hss}}%
    \graphtemp=.5ex
    \advance\graphtemp by 0.915in
    \rlap{\kern 0.236in\lower\graphtemp\hbox to 0pt{\hss $\buy\sep\leave\sep\rec{\Cust_{\it Cu} | \RS_\CustS}$\hss}}%
    \special{pn 8}%
    \special{pa 261 747}%
    \special{pa 236 786}%
    \special{fp}%
    \special{pa 211 747}%
    \special{pa 236 786}%
    \special{fp}%
    \special{pa 236 256}%
    \special{pa 236 786}%
    \special{fp}%
    \graphtemp=.5ex
    \advance\graphtemp by 0.522in
    \rlap{\kern 0.236in\lower\graphtemp\hbox to 0pt{\hss $\enter$~~~~~~~~~~~\hss}}%
    \graphtemp=.5ex
    \advance\graphtemp by 0.522in
    \rlap{\kern 1.417in\lower\graphtemp\hbox to 0pt{\hss $\leave\sep\rec{\Cust_{\it Cu} | \RS_{\CustS}}$\hss}}%
    \special{pa 997 635}%
    \special{pa 1043 647}%
    \special{fp}%
    \special{pa 1013 683}%
    \special{pa 1043 647}%
    \special{fp}%
    \special{pa 610 791}%
    \special{pa 1043 647}%
    \special{fp}%
    \graphtemp=\baselineskip
    \multiply\graphtemp by -1
    \divide\graphtemp by 2
    \advance\graphtemp by .5ex
    \advance\graphtemp by 0.719in
    \rlap{\kern 0.827in\lower\graphtemp\hbox to 0pt{\hss $\buy$~~~~~~~\hss}}%
    \special{pa 640 289}%
    \special{pa 611 253}%
    \special{fp}%
    \special{pa 656 242}%
    \special{pa 611 253}%
    \special{fp}%
    \special{pa 1044 397}%
    \special{pa 611 253}%
    \special{fp}%
    \graphtemp=\baselineskip
    \multiply\graphtemp by -1
    \divide\graphtemp by 2
    \advance\graphtemp by .5ex
    \advance\graphtemp by 0.325in
    \rlap{\kern 0.827in\lower\graphtemp\hbox to 0pt{\hss ~~~~~~~$\leave$\hss}}%
    \graphtemp=.5ex
    \advance\graphtemp by 0.522in
    \rlap{\kern 2.992in\lower\graphtemp\hbox to 0pt{\hss $\rec{\Clerk_{\it Cl} | \RS_\ClerkS}$\hss}}%
    \special{pa 3036 376}%
    \special{pa 2993 393}%
    \special{fp}%
    \special{pa 2996 346}%
    \special{pa 2993 393}%
    \special{fp}%
    \special{pa 2992 650}%
    \special{pa 3189 915}%
    \special{pa 3583 915}%
    \special{pa 3780 522}%
    \special{pa 3583 128}%
    \special{pa 3189 128}%
    \special{pa 2993 393}%
    \special{sp}%
    \graphtemp=.5ex
    \advance\graphtemp by 0.522in
    \rlap{\kern 3.898in\lower\graphtemp\hbox to 0pt{\hss $\serve$\hss}}%
    \hbox{\vrule depth1.043in width0pt height 0pt}%
    \kern 3.898in
  }%
}%
\centerline{\box\graph}
\caption{Labelled transition semantics of customer and clerk}\label{fig:cc}
\vspace{-2ex}
\end{figure}

In order to ensure that the parallel composition
synchronises the \buy-action of the customer with the \serve-action of
the clerk, I apply renaming operators $R_\CustS$ and $R_\ClerkS$ with
$R_\CustS (\buy) = \serves$ and $R_\ClerkS (\serve) = \serves$ and
leaving all other actions unchanged, where
$\serves$ is a joint action of the renamed customer and the renamed clerk.
The total CCSP specification of a store with one clerk and one
customer is $$R_\CustS(\rec{\Cust_{\it Cu} | \RS_\CustS}) \| R_\ClerkS(\rec{\Clerk_{\it Cl} | \RS_\ClerkS})$$
and the relevant part of the labelled transition system of CCSP is displayed below.
\begin{figure}[htb]
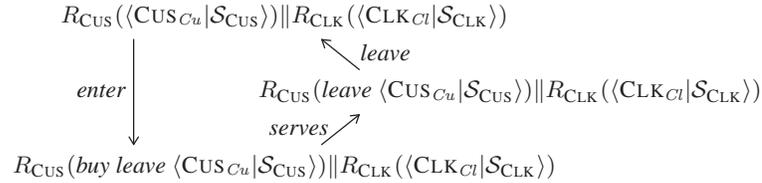

\expandafter\ifx\csname graph\endcsname\relax
   \csname newbox\expandafter\endcsname\csname graph\endcsname
\fi
\ifx\graphtemp\undefined
  \csname newdimen\endcsname\graphtemp
\fi
\expandafter\setbox\csname graph\endcsname
 =\vtop{\vskip 0pt\hbox{%
    \graphtemp=.5ex
    \advance\graphtemp by 0.118in
    \rlap{\kern 0.787in\lower\graphtemp\hbox to 0pt{\hss $R_\CustS(\rec{\Cust_{\it Cu} | \RS_\CustS}) \| R_\ClerkS(\rec{\Clerk_{\it Cl} | \RS_\ClerkS})$\hss}}%
    \graphtemp=.5ex
    \advance\graphtemp by 0.906in
    \rlap{\kern 0.787in\lower\graphtemp\hbox to 0pt{\hss $R_\CustS(\buy\sep\leave\sep\rec{\Cust_{\it Cu} | \RS_\CustS}) \| R_\ClerkS(\rec{\Clerk_{\it Cl} | \RS_\ClerkS})$\hss}}%
    \special{pn 8}%
    \special{pa 25 747}%
    \special{pa 0 786}%
    \special{fp}%
    \special{pa -24 747}%
    \special{pa 0 786}%
    \special{fp}%
    \special{pa 0 236}%
    \special{pa 0 786}%
    \special{fp}%
    \graphtemp=.5ex
    \advance\graphtemp by 0.512in
    \rlap{\kern 0.000in\lower\graphtemp\hbox to 0pt{\hss $\enter$~~~~~~~~~~~\hss}}%
    \graphtemp=.5ex
    \advance\graphtemp by 0.512in
    \rlap{\kern 1.969in\lower\graphtemp\hbox to 0pt{\hss $R_\CustS(\leave\sep\rec{\Cust_{\it Cu} | \RS_{\CustS}}) \| R_\ClerkS(\rec{\Clerk_{\it Cl} | \RS_\ClerkS})$\hss}}%
    \special{pa 1134 636}%
    \special{pa 1180 631}%
    \special{fp}%
    \special{pa 1165 675}%
    \special{pa 1180 631}%
    \special{fp}%
    \special{pa 984 787}%
    \special{pa 1180 631}%
    \special{fp}%
    \graphtemp=.5ex
    \advance\graphtemp by 0.709in
    \rlap{\kern 1.083in\lower\graphtemp\hbox to 0pt{\hss $\serves$~~~~~~~~~~~~~~\hss}}%
    \special{pa 1000 281}%
    \special{pa 985 237}%
    \special{fp}%
    \special{pa 1032 242}%
    \special{pa 985 237}%
    \special{fp}%
    \special{pa 1181 394}%
    \special{pa 985 237}%
    \special{fp}%
    \graphtemp=.5ex
    \advance\graphtemp by 0.315in
    \rlap{\kern 1.083in\lower\graphtemp\hbox to 0pt{\hss ~~~~~~~~~~~~~~~$\leave$\hss}}%
    \hbox{\vrule depth1.024in width0pt height 0pt}%
    \kern 2.205in
  }%
}%
\centerline{\box\graph}
\caption{Labelled transition semantics of the 1-customer 1-clerk store}\label{fig:store1}
\vspace{-2ex}
\end{figure}

\noindent
One possible behaviour of this system is the sequence of actions
$\enter \sep \serves \sep \leave$ $\enter$, followed by eternal stagnation.
This behaviour is ruled out by the progress assumption of \SSect{inevitability}.
The only behaviour compatible with this assumption is the infinite
sequence of actions $(\enter \sep \serves \sep \leave)^\infty$.

To model a store with two customers (A and B) and 2 clerks (I and II),
I introduce a relational renaming for each of them, defined by
$$\begin{array}{c@{\quad\!\!}l@{\,=\,}l@{}c}
R_{A}(\enter)\mathbin=\ienters & R_{A}(\buy)&\{\ik,\il\} & R_{A}(\leave)\mathbin=\ileaves \\
R_{B}(\enter)\mathbin=\jenters & R_{B}(\buy)&\{\jk,\jl\} & R_{B}(\leave)\mathbin=\jleaves \\
& R_{\rm I}(\serve)&\{\ik,\jk\} \\
& R_{\rm II}(\serve)&\{\il,\jl\}.
\end{array}$$
\hypertarget{store2}{The CCSP specification of a store with two clerks and two
customers is
\[
\big(\!R_A(\!\rec{\Cust_{\it Cu} | \RS_\CustS}\!) \| R_B(\!\rec{\Cust_{\it Cu} | \RS_\CustS}\!)\!\big) \|
\big(\!R_{\rm I}(\!\rec{\Clerk_{\it Cl} | \RS_\ClerkS}\!) \| R_{\rm II}(\!\rec{\Clerk_{\it Cl} | \RS_\ClerkS}\!)\!\big)
\]}
and the part of the labelled transition system of CCSP reachable from
that process has $3\times 3 \times 1 \times 1 = 9$ states and $6 \times 4 = 24$ transitions.
\end{example}

\section{Petri nets}\label{sec:nets}

A {\em multiset} over a set $\SX$ is a function $\A\!:\SX \rightarrow\IN$, i.e.\ $\A\in \IN^{\SX}$;
let $|\A| := \sum_{x\in X}\A(x)$;
$x \mathbin\in \SX$ is an \emph{element of} $\A$, notation $x \mathbin\in \A$, iff $\A(x) > 0$.\\
The function $\emptyset\!:\!\SX\mathbin\rightarrow\IN$, given by $\emptyset(x)\mathbin{:=}0$
for all $x \mathbin\in \SX$, is the \emph{empty} multiset over $\SX$.
For multisets $\A$ and $\B$ over $\SX$ one writes $\A \leq \B$ iff
 \mbox{$\A(x) \leq \B(x)$} for all $x \mathbin\in \SX$;
\\ $\A\cap \B$ denotes the multiset over $\SX$ with $(\A\cap \B)(x):=\textrm{min}(\A(x), \B(x))$,
\\ $\A + \B$ denotes the multiset over $\SX$ with $(\A + \B)(x):=\A(x)+\B(x)$; and
\\ the multiset $\A - \B$ is only defined if $\B\leq \A$ and then
$(\A - \B)(x):=\A(x)-\B(x)$.
A multiset $\A$ with $\A(x)\leq 1$ for all $x$ is identified with
the (plain) set $\{x\mid \A(x)\mathbin=1\}$.
The construction $\A:=\{f(x_1,...,x_n) \mid x_i\mathbin\in \B_i\}$
of a set $\A$ out of sets $\B_i$ ($i\mathbin=1,...,n$)
generalises naturally to multisets $\A$ and $\B_i$,
taking the multiplicity $\A(x)$ of an element $x$ to be
$\sum_{f(x_1,...,x_n)=x}\B_1(x_1) \cdot ... \cdot \B_n(x_n)$.

\begin{definition}\rm
  A (\emph{typed}) \emph{Petri net} is a tuple
  $N = (S, T, F, M_0, A, \ell)$ with
  \begin{list}{{\bf --}}{\leftmargin 18pt
                        \labelwidth\leftmargini\advance\labelwidth-\labelsep
                        \topsep 0pt \itemsep 0pt \parsep 0pt}
    \item $S$ and $T$ disjoint sets (of \emph{places} and \emph{transitions}),
    \item $F: ((S \times T) \cup (T \times S)) \rightarrow \IN$
      (the \emph{flow relation} including \emph{arc weights}),
    \item $M_0 : S \rightarrow \IN$ (the \emph{initial marking}),
    \item $A$ a set of \emph{actions}, the \emph{type} of the net, and
    \item \plat{$\ell: T \rightarrow A$} (the \emph{labelling function}).
    \pagebreak[3]
  \end{list}
\end{definition}

\noindent
Petri nets are depicted by drawing the places as circles and the
transitions as boxes, containing their label.
Identities of places and transitions are displayed next to the net element.
For $x,y \mathbin\in S\cup T$ there are $F(x,y)$
arrows (\emph{arcs}) from $x$ to $y$. When a Petri net represents a
concurrent system, a global state of this system is given as a \emph{marking},
a multiset $M$ of places, depicted by placing $M(s)$ dots (\emph{tokens}) in each place $s$. 
The initial state is $M_0$.

The behaviour of a Petri net is defined by the possible moves between
markings $M$ and $M'$, which take place when a transition $\tu$ \emph{fires}.  In that case,
$\tu$ consumes $F(s,\tu)$ tokens from each 
place $s$.  Naturally, this can happen only if $M$ makes all these
tokens available in the first place.  Moreover, $\tu$ produces $F(\tu,s)$ tokens
in each place $s$.  \Def{firing} formalises this notion of behaviour.

\begin{definition}\label{df:firing}\rm
Let $N = (S, T, F, M_0, A, \ell)$ be a Petri net and $x\mathbin\in S\cup T$.
The multisets $\precond{x},~\postcond{x}{:}\, S \mathord\cup T \rightarrow
\IN$ are given by $\precond{x}(y)=F(y,x)$ and
$\postcond{x}(y)=F(x,y)$ for all $y \mathbin\in S\mathord\cup T$;
for $t\mathbin\in T$, the elements of $\precond{\tu}$ and $\postcond{\tu}$ are
called \emph{pre-} and \emph{postplaces} of $\tu$, respectively.
Transition $\tu\mathbin\in T$ is \emph{enabled} from the marking
$M\mathbin\in\IN^S$---notation $M[\tu\rangle$---if $\precond{\tu} \leq M$.
In that case firing $\tu$ yields the marking 
$M':=M-\precond{\tu}+\postcond{\tu}$\linebreak[3]---notation $M[\tu\rangle M'$.
\end{definition}
A \emph{path} $\pi$ of a Petri net $N$ is an alternating sequence $M_0 \tu_1 M_1 \tu_2 M_2 \tu_3 \dots$ of markings and
transitions, starting from the initial marking $M_0$ and either being infinite or ending in a
marking $M_n$, such that $M_k [\tu_k\rangle M_{k+1}$ for all $k \,(\mathord< n)$.
A marking is \emph{reachable} if it occurs in such a path.
The Petri net $N$ is
\emph{safe} if all reachable markings $M$ are plain sets, meaning that $M(s)\leq 1$ for all places $s$.
It has \emph{bounded parallelism} \cite{GV87} if there is no reachable marking $M$
and infinite multiset of transitions $U$ such that $\sum_{t \in U}\precond{t}\leq M$.
In this paper I consider Petri nets with bounded parallelism only, and call them \emph{nets}.

\section{An operational Petri net semantics of CCSP}\label{sec:operational}

This section recalls the operational Petri net semantics of CCSP, given by Olderog \cite{Old87,Old91}.
It associates a net $\denote{P}$ with each closed CCSP expression $P$.

The standard operational semantics of CCSP, presented in \Sect{CCSP}, yields one big labelled
transition system for the entire language.\footnote{A \emph{labelled transition system} (LTS)
  is given by a set $S$ of \emph{states} and a \emph{transition relation}
  \mbox{$T\subseteq S\times \Lab \times S$} for some set of labels $\Lab$.
  The LTS generated by CCSP has $S:=\T_{\rm CCSP}$, $\Lab:=Act$ and $T:=\mathord{\rightarrow}$.
}
Each individual closed CCSP expression $P$ appears as a state in
this LTS\@.  If desired, a \emph{process graph}---an LTS enriched with an initial state---for $P$
can be extracted from this system-wide LTS by appointing $P$ as the initial state, and optionally
deleting all states and transitions not reachable from $P$. In the same vein, an operational Petri
net semantics yields one big Petri net for the entire language, but without an initial marking.
I call such a Petri net {\em unmarked}. Each
process $P\mathbin\in\T_{\rm CCSP}$ corresponds with a marking
$dex(P)$ of that net. If desired, a Petri net $\denote{P}$
for $P$ can be extracted from this system-wide net by appointing $dex(P)$ as its initial marking,
taking the type of $\denote{P}$ to be $\alpha(P)$,
and optionally deleting all places and transitions not reachable from $dex(P)$.

The set $\SC_{\rm CCSP}$ of places in the net is the smallest set including:
\begin{center}
\begin{tabular}{@{}l@{\quad}l@{\qquad\quad}l@{\quad}l@{\qquad\quad}l@{\quad}l@{}}
$0_A$ & \emph{inaction} &
$a \E$  & \emph{prefixing}&
$\mu+\nu$  & \emph{choice} \\
$\mu\|_A$ & \emph{left parallel component}&
$_A\|\mu$ & \emph{right component} &
$R(\mu)$ &  \emph{renaming} \\
\end{tabular}
\end{center}
\noindent for $A\subseteq Act$, $\E\mathbin\in\T_{\rm CCSP}$, $a\mathbin\in Act$,
$\mu,\nu\mathbin\in\SC_{\rm CCSP}$ and renamings $R$.
The mapping $dex:\T_{\rm CCSP} \rightarrow \pow(\SC_{\rm CCSP})$ decomposing and expanding a process
expression into a set of places is inductively defined by:\vspace{-2ex}
\[
\begin{array}{@{}l@{~=~}l@{\qquad\qquad\qquad}l@{~=~}l@{}}
\multicolumn{2}{c}{} &
dex(0_A) & \{0_A\} \\
dex(a P)  & \{a P\} &
dex(R(P)) & R(dex(P)) \\
dex(P+Q) & dex(P) + dex(Q) &
dex(\rec{X_A|\RS}) & dex(\rec{\RS_{X_A}|\RS}) \\
dex(P\|Q) & \multicolumn{3}{l}{dex(P)\|_A~ \cup ~_A\|dex(Q) ~\mbox{where}~ A=\alpha(P)\cap\alpha(Q).}
\end{array}
\]
Here $H\|_A$, $_A\|H$, $R(H)$ and $H+K$ for $H,K\mathbin\subseteq \SC_{\rm CCSP}$
are defined element by element; e.g.\ $R(H) = \{R(\mu) \mid \mu\mathbin\in H\}$.
The binding matters, so that $(_A\|H)\|_B\mathbin{\not=}{_A\|}(H\|_B)$.
Since {\we} deal with guarded recursion only, $dex$ is well-defined.
\begin{table}[bt]
\caption{Operational Petri net semantics of CCSP}
\label{tab:PN-CCSP}
\vspace{-2ex}
\normalsize
\begin{center}
\framebox{$\begin{array}{c@{\qquad}l}
\multicolumn{2}{c}{\{aP\} \ar{a} dex(P)} \\[2ex]
\displaystyle\frac{H \ar{a} J}{R(H) \ar{b} R(J)} ~~((a,b)\mathbin\in R) &
\qquad\displaystyle \frac{H \ar{a} J}{H\|_A \ar{a} J\|_A}~~(a\mathbin{\notin}A) \\[4ex]
\displaystyle \frac{H\dcup K \ar{a} J}{H\cup(K+dex(Q)) \ar{a} J} &
\displaystyle \frac{H \ar{a} J \qquad K \ar{a} L}{H\|_A \cup {_A\|}K  \ar{a} J\|_A \cup {_A\|}L}~~
                                 (a\mathbin\in A) \\[4ex]
\displaystyle \frac{H\dcup K \ar{a} J}{H\cup(dex(P)+K) \ar{a} J} &
\qquad\displaystyle \frac{H \ar{a} J}{_A\|H \ar{a} {_A\|}J}~~(a\mathbin{\notin}A) \\[4ex]
\end{array}$}
\end{center}
\vspace{-2ex}
\end{table}

Following \cite{Old91}, {\we} construct the unmarked Petri net $(S,T,F,Act,\ell)$ of CCSP
with $S:=\SC_{\rm CCSP}$, specifying the triple $(T,F,\ell)$ as a ternary relation
$\mathord{\rightarrow} \subseteq \IN^S\times Act\times \IN^S$.
An element \plat{$H \ar{a} J$} of this relation denotes a transition
$t\mathbin\in T$ with $\ell(t)\mathbin=a$
such that $\precond{t}\mathbin=H$ and $\postcond{t}\mathbin=J$.
The transitions \plat{$H\ar{\alpha}J$} are derived from the rules of \Tab{PN-CCSP}.

Note that there is no rule for recursion. The transitions of a recursive process $\rec{\hspace{-1pt}X_A|\RS}$
are taken care of indirectly by the decomposition $dex(\rec{X_A|\RS}) = dex(\rec{\RS_{X_A}|\RS})$,
which expands the decomposition of a recursive call into a decomposition of
an expression in which each recursive call is guarded by an action prefix.

\begin{example}
The Petri net semantics of the \hyperlink{store2}{2-customer 2-clerk store} from
\Sect{CCSP} is displayed in \Fig{store2}. It is more compact than
the 9-state 24-transition labelled transition system. The name of the
bottom-most place is 
$_{\it Ser\!}\|\, {}_{\emptyset\!}\|\, R_{\rm II}(\serve\sep\rec{\Clerk_{\it Cl} | \RS_\ClerkS})$
where $\it Ser$ is the alphabet $\{\ik,\jk,\il,\jl\}$.

A progress assumption, as discussed in \SSect{inevitability}, disallows
runs that stop after finitely many actions. So in each run some of
the actions from {\it Ser} will occur infinitely often. When assuming
strong fairness, each of those actions will occur infinitely often.
When assuming only weak fairness, it is possible that $\il$ and $\jl$
will never occur, as long as $\ik$ and $\jk$ each occur infinitely often,
for in such a run the actions $\il$ and $\jl$ are not enabled in every
state (from some point onwards).
However, it is not possible that $\jk$ and $\jl$ never occur, because
in such a run, from some point onwards, the action $\jk$ is enabled
in every state.

When assuming justness but not weak fairness, a run that bypasses any two
serving actions is possible, but a run that bypasses $\jk,\il$
and $\jl$ is excluded, because in such a run, from some point onwards,
the action $\jl$ is perpetually enabled, in the sense that both tokens
in its preplaces never move away.
\begin{figure}[hb]
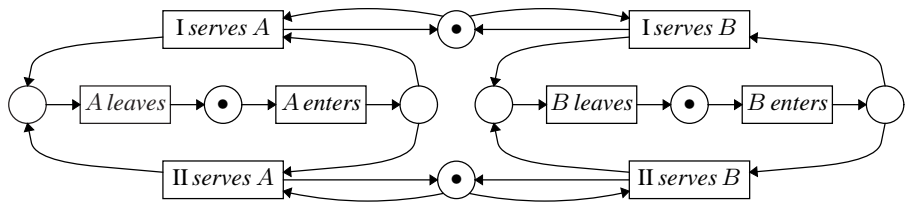

\expandafter\ifx\csname graph\endcsname\relax
   \csname newbox\expandafter\endcsname\csname graph\endcsname
\fi
\ifx\graphtemp\undefined
  \csname newdimen\endcsname\graphtemp
\fi
\expandafter\setbox\csname graph\endcsname
 =\vtop{\vskip 0pt\hbox{%
    \special{pn 8}%
    \special{ar 98 503 98 98 0 6.28319}%
    \special{pa 374 601}%
    \special{pa 846 601}%
    \special{pa 846 405}%
    \special{pa 374 405}%
    \special{pa 374 601}%
    \special{fp}%
    \graphtemp=.5ex
    \advance\graphtemp by 0.503in
    \rlap{\kern 0.610in\lower\graphtemp\hbox to 0pt{\hss $\ileaves$\hss}}%
    \special{sh 1.000}%
    \special{pn 1}%
    \special{pa 335 478}%
    \special{pa 374 503}%
    \special{pa 335 528}%
    \special{pa 335 478}%
    \special{fp}%
    \special{pn 8}%
    \special{pa 197 503}%
    \special{pa 335 503}%
    \special{fp}%
    \special{ar 1122 503 98 98 0 6.28319}%
    \graphtemp=.5ex
    \advance\graphtemp by 0.503in
    \rlap{\kern 1.122in\lower\graphtemp\hbox to 0pt{\hss $\bullet$\hss}}%
    \special{sh 1.000}%
    \special{pn 1}%
    \special{pa 984 478}%
    \special{pa 1024 503}%
    \special{pa 984 528}%
    \special{pa 984 478}%
    \special{fp}%
    \special{pn 8}%
    \special{pa 846 503}%
    \special{pa 984 503}%
    \special{fp}%
    \special{pa 1398 601}%
    \special{pa 1870 601}%
    \special{pa 1870 405}%
    \special{pa 1398 405}%
    \special{pa 1398 601}%
    \special{fp}%
    \graphtemp=.5ex
    \advance\graphtemp by 0.503in
    \rlap{\kern 1.634in\lower\graphtemp\hbox to 0pt{\hss $\ienters$\hss}}%
    \special{sh 1.000}%
    \special{pn 1}%
    \special{pa 1358 478}%
    \special{pa 1398 503}%
    \special{pa 1358 528}%
    \special{pa 1358 478}%
    \special{fp}%
    \special{pn 8}%
    \special{pa 1220 503}%
    \special{pa 1358 503}%
    \special{fp}%
    \special{ar 2146 503 98 98 0 6.28319}%
    \special{sh 1.000}%
    \special{pn 1}%
    \special{pa 2008 478}%
    \special{pa 2047 503}%
    \special{pa 2008 528}%
    \special{pa 2008 478}%
    \special{fp}%
    \special{pn 8}%
    \special{pa 1870 503}%
    \special{pa 2008 503}%
    \special{fp}%
    \special{pa 807 208}%
    \special{pa 1437 208}%
    \special{pa 1437 11}%
    \special{pa 807 11}%
    \special{pa 807 208}%
    \special{fp}%
    \graphtemp=.5ex
    \advance\graphtemp by 0.109in
    \rlap{\kern 1.122in\lower\graphtemp\hbox to 0pt{\hss $\ik$\hss}}%
    \special{sh 1.000}%
    \special{pn 1}%
    \special{pa 1474 177}%
    \special{pa 1437 149}%
    \special{pa 1478 127}%
    \special{pa 1474 177}%
    \special{fp}%
    \special{pn 8}%
    \special{pa 2146 405}%
    \special{pa 2106 208}%
    \special{pa 1441 149}%
    \special{sp}%
    \special{sh 1.000}%
    \special{pn 1}%
    \special{pa 131 371}%
    \special{pa 98 405}%
    \special{pa 82 361}%
    \special{pa 131 371}%
    \special{fp}%
    \special{pn 8}%
    \special{pa 99 401}%
    \special{pa 138 208}%
    \special{pa 807 149}%
    \special{sp}%
    \special{pa 807 995}%
    \special{pa 1437 995}%
    \special{pa 1437 798}%
    \special{pa 807 798}%
    \special{pa 807 995}%
    \special{fp}%
    \graphtemp=.5ex
    \advance\graphtemp by 0.897in
    \rlap{\kern 1.122in\lower\graphtemp\hbox to 0pt{\hss $\il$\hss}}%
    \special{sh 1.000}%
    \special{pn 1}%
    \special{pa 1478 879}%
    \special{pa 1437 857}%
    \special{pa 1474 829}%
    \special{pa 1478 879}%
    \special{fp}%
    \special{pn 8}%
    \special{pa 2146 601}%
    \special{pa 2106 798}%
    \special{pa 1441 857}%
    \special{sp}%
    \special{sh 1.000}%
    \special{pn 1}%
    \special{pa 82 645}%
    \special{pa 98 601}%
    \special{pa 131 635}%
    \special{pa 82 645}%
    \special{fp}%
    \special{pn 8}%
    \special{pa 99 605}%
    \special{pa 138 798}%
    \special{pa 807 857}%
    \special{sp}%
    \special{ar 2539 503 98 98 0 6.28319}%
    \special{pa 2815 601}%
    \special{pa 3287 601}%
    \special{pa 3287 405}%
    \special{pa 2815 405}%
    \special{pa 2815 601}%
    \special{fp}%
    \graphtemp=.5ex
    \advance\graphtemp by 0.503in
    \rlap{\kern 3.051in\lower\graphtemp\hbox to 0pt{\hss $\jleaves$\hss}}%
    \special{sh 1.000}%
    \special{pn 1}%
    \special{pa 2776 478}%
    \special{pa 2815 503}%
    \special{pa 2776 528}%
    \special{pa 2776 478}%
    \special{fp}%
    \special{pn 8}%
    \special{pa 2638 503}%
    \special{pa 2776 503}%
    \special{fp}%
    \special{ar 3563 503 98 98 0 6.28319}%
    \graphtemp=.5ex
    \advance\graphtemp by 0.503in
    \rlap{\kern 3.563in\lower\graphtemp\hbox to 0pt{\hss $\bullet$\hss}}%
    \special{sh 1.000}%
    \special{pn 1}%
    \special{pa 3425 478}%
    \special{pa 3465 503}%
    \special{pa 3425 528}%
    \special{pa 3425 478}%
    \special{fp}%
    \special{pn 8}%
    \special{pa 3287 503}%
    \special{pa 3425 503}%
    \special{fp}%
    \special{pa 3839 601}%
    \special{pa 4311 601}%
    \special{pa 4311 405}%
    \special{pa 3839 405}%
    \special{pa 3839 601}%
    \special{fp}%
    \graphtemp=.5ex
    \advance\graphtemp by 0.503in
    \rlap{\kern 4.075in\lower\graphtemp\hbox to 0pt{\hss $\jenters$\hss}}%
    \special{sh 1.000}%
    \special{pn 1}%
    \special{pa 3799 478}%
    \special{pa 3839 503}%
    \special{pa 3799 528}%
    \special{pa 3799 478}%
    \special{fp}%
    \special{pn 8}%
    \special{pa 3661 503}%
    \special{pa 3799 503}%
    \special{fp}%
    \special{ar 4587 503 98 98 0 6.28319}%
    \special{sh 1.000}%
    \special{pn 1}%
    \special{pa 4449 478}%
    \special{pa 4488 503}%
    \special{pa 4449 528}%
    \special{pa 4449 478}%
    \special{fp}%
    \special{pn 8}%
    \special{pa 4311 503}%
    \special{pa 4449 503}%
    \special{fp}%
    \special{pa 3248 208}%
    \special{pa 3878 208}%
    \special{pa 3878 11}%
    \special{pa 3248 11}%
    \special{pa 3248 208}%
    \special{fp}%
    \graphtemp=.5ex
    \advance\graphtemp by 0.109in
    \rlap{\kern 3.563in\lower\graphtemp\hbox to 0pt{\hss $\jk$\hss}}%
    \special{sh 1.000}%
    \special{pn 1}%
    \special{pa 3915 177}%
    \special{pa 3878 149}%
    \special{pa 3919 127}%
    \special{pa 3915 177}%
    \special{fp}%
    \special{pn 8}%
    \special{pa 4587 405}%
    \special{pa 4547 208}%
    \special{pa 3882 149}%
    \special{sp}%
    \special{sh 1.000}%
    \special{pn 1}%
    \special{pa 2572 371}%
    \special{pa 2539 405}%
    \special{pa 2523 361}%
    \special{pa 2572 371}%
    \special{fp}%
    \special{pn 8}%
    \special{pa 2540 401}%
    \special{pa 2579 208}%
    \special{pa 3248 149}%
    \special{sp}%
    \special{pa 3248 995}%
    \special{pa 3878 995}%
    \special{pa 3878 798}%
    \special{pa 3248 798}%
    \special{pa 3248 995}%
    \special{fp}%
    \graphtemp=.5ex
    \advance\graphtemp by 0.897in
    \rlap{\kern 3.563in\lower\graphtemp\hbox to 0pt{\hss $\jl$\hss}}%
    \special{sh 1.000}%
    \special{pn 1}%
    \special{pa 3919 879}%
    \special{pa 3878 857}%
    \special{pa 3915 829}%
    \special{pa 3919 879}%
    \special{fp}%
    \special{pn 8}%
    \special{pa 4587 601}%
    \special{pa 4547 798}%
    \special{pa 3882 857}%
    \special{sp}%
    \special{sh 1.000}%
    \special{pn 1}%
    \special{pa 2523 645}%
    \special{pa 2539 601}%
    \special{pa 2572 635}%
    \special{pa 2523 645}%
    \special{fp}%
    \special{pn 8}%
    \special{pa 2540 605}%
    \special{pa 2579 798}%
    \special{pa 3248 857}%
    \special{sp}%
    \special{ar 2343 109 98 98 0 6.28319}%
    \graphtemp=.5ex
    \advance\graphtemp by 0.109in
    \rlap{\kern 2.343in\lower\graphtemp\hbox to 0pt{\hss $\bullet$\hss}}%
    \special{sh 1.000}%
    \special{pn 1}%
    \special{pa 1481 66}%
    \special{pa 1437 50}%
    \special{pa 1470 17}%
    \special{pa 1481 66}%
    \special{fp}%
    \special{pn 8}%
    \special{ar 1879 1969 1969 1969 -1.777375 -1.369439}%
    \special{sh 1.000}%
    \special{pn 1}%
    \special{pa 2205 84}%
    \special{pa 2244 109}%
    \special{pa 2205 134}%
    \special{pa 2205 84}%
    \special{fp}%
    \special{pn 8}%
    \special{pa 1437 109}%
    \special{pa 2205 109}%
    \special{fp}%
    \special{sh 1.000}%
    \special{pn 1}%
    \special{pa 3215 17}%
    \special{pa 3248 50}%
    \special{pa 3204 66}%
    \special{pa 3215 17}%
    \special{fp}%
    \special{pn 8}%
    \special{ar 2806 1969 1969 1969 -1.772154 -1.364218}%
    \special{sh 1.000}%
    \special{pn 1}%
    \special{pa 2480 134}%
    \special{pa 2441 109}%
    \special{pa 2480 84}%
    \special{pa 2480 134}%
    \special{fp}%
    \special{pn 8}%
    \special{pa 3248 109}%
    \special{pa 2480 109}%
    \special{fp}%
    \special{ar 2343 897 98 98 0 6.28319}%
    \graphtemp=.5ex
    \advance\graphtemp by 0.897in
    \rlap{\kern 2.343in\lower\graphtemp\hbox to 0pt{\hss $\bullet$\hss}}%
    \special{sh 1.000}%
    \special{pn 1}%
    \special{pa 3204 940}%
    \special{pa 3248 956}%
    \special{pa 3215 989}%
    \special{pa 3204 940}%
    \special{fp}%
    \special{pn 8}%
    \special{ar 2806 -961 1969 1969 1.364218 1.772154}%
    \special{sh 1.000}%
    \special{pn 1}%
    \special{pa 2480 922}%
    \special{pa 2441 897}%
    \special{pa 2480 872}%
    \special{pa 2480 922}%
    \special{fp}%
    \special{pn 8}%
    \special{pa 3248 897}%
    \special{pa 2480 897}%
    \special{fp}%
    \special{sh 1.000}%
    \special{pn 1}%
    \special{pa 1470 989}%
    \special{pa 1437 956}%
    \special{pa 1481 940}%
    \special{pa 1470 989}%
    \special{fp}%
    \special{pn 8}%
    \special{ar 1879 -961 1969 1969 1.369439 1.777375}%
    \special{sh 1.000}%
    \special{pn 1}%
    \special{pa 2205 872}%
    \special{pa 2244 897}%
    \special{pa 2205 922}%
    \special{pa 2205 872}%
    \special{fp}%
    \special{pn 8}%
    \special{pa 1437 897}%
    \special{pa 2205 897}%
    \special{fp}%
    \hbox{\vrule depth1.006in width0pt height 0pt}%
    \kern 4.685in
  }%
}%
\centerline{\box\graph}
\caption{Petri net semantics of the 2-customer 2-clerk store}\label{fig:store2}
\vspace{-2ex}
\end{figure}

\end{example}

Olderog \cite{Old87,Old91} shows that the Petri net $\denote{P}$
associated to a closed CCSP
expression $P$ is safe, and that all its reachable markings are
finite; the latter implies that it has bounded parallelism.
The following result, from \cite{Old87,Old91}, shows that the standard
interleaving semantics of CCSP is retrievable from the net semantics;
it establishes a strong bisimulation relating any CCSP expression
(seen as a state in a labelled transition system) with its
interpretation as a marking in the Petri net of CCSP\@.

\begin{theorem}\rm
There exists a relation $\Bi$ between closed CCSP expressions and
markings in the unmarked Petri net of CCSP, such that
\vspace{-1ex}
\begin{itemize}
\item $P \mathrel{\Bi} dex(P)$ for each closed, well-typed CCSP expression with guarded recursion,
\item if $P\Bi M$ and $P\ar{a}P'$ then there is a marking $M'$ and
  transition $t$ with $\ell(t)=a$, $M[t\rangle M'$ and $P\Bi M'$, and
\item if $P\Bi M$ and $M[t\rangle M'$ then there is CCSP process $P'$
  with $P\ar{\ell(t)}P'$ and $P\Bi M'$.
\end{itemize}
\end{theorem}

To formalise the concurrency requirement for his net semantics Olderog
defines for each $n$-ary CCSP operator $op$ an $n$-ary operation
$op_{\cal N}$ on safe Petri nets, inspired by proposals from \cite{GM84,Wi84,GV87},
and requires that
\hypertarget{concurrency_requirement}{%
$$\begin{array}{c@{\qquad}r@{~\approx~}l}
(1)& \denote{op(P_1,\dots,P_n)} & op_{\cal N}(\denote{P_1},\dots,\denote{P_n}) \\
(2)& \denote{\rec{X_A|\RS}} & \denote{\rec{\RS_{X_A}|\RS}}
\end{array}$$}
for a suitable relation $\approx$.
In fact, (2) turns out to hold taking for $\approx$ the identity relation.
He establishes (1) taking for $\approx$ a relation he calls \emph{strong bisimilarity}, whose
definition will be recalled in \Sect{strong}. When a relation $\equiv$ includes $\approx$, and
(1) holds for $\approx$, then it also holds for $\equiv$.

The operations $op_{\cal N}$ (i.e.\ $({0_A})_{\cal N}$ for $A \mathbin\subseteq Act$,
$a_{\cal N}$ for $a \mathbin\in Act$, $R_{\cal N}$ for $R \mathbin\subseteq Act \times Act$,
$\|_{\cal N}$ and $+_{\cal N}$) are defined only up to isomorphism, but
this is no problem as isomorphic nets are strongly bisimilar.
The definition is recalled below---it generalises verbatim to non-safe nets,
except that $+_{\cal N}$ is defined only for nets whose initial markings are nonempty plain sets.

\begin{definition}\rm$\!$\cite{Old91}\label{df:operators}
The net $0_A$ has type $A$ and consists of a single place, initially marked:
$({0_A})_{\cal N} := (\{0_A\},\emptyset,\emptyset,\{0_A\},A,\emptyset)$.

Given a net $N=(S,T,F,M,A,\ell)$ and $a\mathbin\in Act$, take $s_0,t_a \not\in S \cup T$.
Then the net $a_{\cal N}N$ is obtained from $N$ by the addition of the
fresh place $s_0$ and the fresh transition $t_a$, labelled $a$, such that
$\precond{t_a}=\{s_0\}$ and $\postcond{t_a}=M$. The type of $a_{\cal N}N$
will be $A\cup\{a\}$ and the initial marking $\{s_0\}$.

Given a net $N=(S,T,F,M,A,\ell)$ and a renaming operator $R(\_)$, the net $R_{\cal N}(N)$
has type $R(A):=\{b\mathbin\in Act \mid \exists a \mathbin\in A, (a,b)\mathbin\in R\}$,
the same places and initial marking as $N$, and transitions
$t_b$ for each $t\mathbin\in T$ and $b\mathbin\in Act$ with $(\ell(t),b)\mathbin\in R$.
One has $\precond{t_b}:=\precond{t}$, $\postcond{t_b}:=\postcond{t}$,
and the label of $t_b$ will be $b$.

Given two nets $N_i=(S_i,T_i,F_i,M_i,A_i,\ell_i)$ ($i\mathbin=1,2$), their parallel composition 
$N_1\|_{\cal N}N_2=(S,T,F,M,A,\ell)$ is obtained
from the disjoint union of $N_1$ and $N_2$ by the omission of all transitions $t$ of $T_1\dcup T_2$
with $\ell(t)\mathbin\in A_1\cap A_2$, and the addition of fresh transitions $(t_1,t_2)$ for
all pairs $t_i\mathbin\in T_i$ ($i\mathbin=1,2$) with $\ell_1(t_1)=\ell_2(t_2)\in A_1 \cap A_2$.
Take $\precond(t_1,t_2) = \precond{t_1} + \precond{t_2}$,
$\postcond{(t_1,t_2)} = \postcond{t_1} + \postcond{t_2}$,
$\ell(t_1,t_2)=\ell(t_1)$, and $A:=A_1 \cup A_2$.

Given nets $N_i=(S_i,T_i,F_i,M_i,A_i,\ell_i)$ with $M_i\mathbin{\neq}\emptyset$ a plain set ($i\mathbin=1,2$),
the net $N_1+_{\cal N}N_2$ with type $A_1\cup A_2$ is obtained
from the disjoint union of $N_1$ and $N_2$ by the addition of the set
of fresh places $M_1\times M_2$---this set will be the initial marking of
$N_1\mathord{+_{\cal N}}N_2$---and the addition of fresh transitions
$t_i^K$ for any $t_i\mathbin\in T_i$ and $\emptyset \mathbin{\neq} K \mathbin{\leq} \precond{t_i}\cap M_i$.
$\ell(t^K_i)\mathbin=\ell_i(t)$, $\precond{t^K_1}\mathbin= \precond{t_1}{-}K+(K\times M_2)$,
$\precond{t^K_2}\mathbin= \precond{t_2}{-}K+(M_1\times K)$ and $\postcond{(t_i^K)}\mathbin=\postcond{t_i}$.
\end{definition}

\section{Structure preserving bisimulation equivalence}\label{sec:sp-bis}

This section presents structure preserving bisimulation equivalence on nets.

\begin{definition}\rm
Given two nets $N_i\mathbin=(S_i,\!T_i,\!F_i,\!M_i,\!A_i,\!\ell_i)$, a \emph{link} is a pair $(s_1,s_2)\mathord\in S_1{\times} S_2$
of places.
A \emph{linking} $l\mathbin\in \IN^{S_1\times S_1}$ is a multiset of links;
it can be seen as a pair of markings with a bijection between them.
Let $\pi_i(l)\mathbin\in\IN^{S_i}$ be these markings, given by
$\pi_1(l)(s_1)=\sum_{s_2\in S_2}l(s_1,s_2)$ for all $s_1\mathbin\in S_1$ and $\pi_2(l)(s_2)=\sum_{s_1\in S_1}l(s_1,s_2)$
for all $s_2\mathbin\in S_2$.\linebreak[3]
A \emph{structure preserving bisimulation} (\emph{sp-bisimulation}) is a set $\Bi$ of linkings, such that
\begin{itemize}\vspace{-1ex}
\item if $c\leq l\mathbin\in \Bi$ and $\pi_1(c)\mathbin=\precond{t_1}$ for $t_1\mathbin\in T_1$ then there
  are a transition $t_2\mathbin\in T_2$ with $\ell(t_2)=\ell(t_1)$ and $\pi_2(c)=\precond{t_2}$,
  and a linking $\bar c$ such that  $\pi_1(\bar c)=\postcond{t_1}$, $\pi_2(\bar c)=\postcond{t_2}$
  and $\bar l:=l-c+\bar c\in\Bi$.
\item if $c\leq l\mathbin\in \Bi$ and $\pi_2(c)=\precond{t_2}$ then there are a $t_1$
  and a $\bar c$ with the same properties.
\vspace{-1ex}
\end{itemize}
$N_1$ and $N_2$ are \emph{structure preserving bisimilar}, notation $N_1\mathbin{\bis{sp}}N_2$,
if $A_1\mathbin= A_2$ and there is a linking $l$ in a structure preserving bisimulation with $M_1\mathbin=\pi_1(l)$ and
$M_2 \mathbin=\pi_2(l)$.
\end{definition}
Note that if $\Bi$ is an sp-bisimulation, then so is its downward closure
$\{k\mid \exists l\mathbin\in \Bi.\, k\leq l\}$. 
Moreover, if $\Bi$ is an sp-bisimulation between two nets,
then the set of those linkings $l\mathbin\in\Bi$ for which
$\pi_1(l)$ and $\pi_2(l)$ are reachable markings is also an sp-bisimulation.

If $\cal B$ is a set of a links, let $\overline{\cal B}$ be the set of \emph{all}
linkings that are multisets over $\cal B$.

\begin{proposition}\rm
Structure preserving bisimilarity is an equivalence relation.
\end{proposition}
\begin{proof}
The relation $\overline{\it Id}$, with ${\it Id}$ the identity
relation on places, is an sp-bisimulation, showing that $N\bis{sp}N$ for any net $N$.

Given an sp-bisimulation $\Bi$, also $\{l^{-1}\mid l\mathbin\in \Bi\}$ is an sp-bisimulation, showing symmetry
of $\bis{sp}$.

Given linkings $h\mathbin\in \IN^{S_1\times S_3}$, $k\mathbin\in \IN^{S_1\times S_2}$ and $l\mathbin\in \IN^{S_2\times S_3}$,
write $h \in k;l$ if there is a multiset
$m\mathbin\in \IN^{S_1\times S_2\times S_3}$ of triples of places, with
$k(s_1,s_2) = \sum_{s_3\in S}  m(s_1,s_2,s_3)$,
$l(s_2,s_3) = \sum_{s_1\in S}  m(s_1,s_2,s_3)$ and
$h(s_1,s_3) = \sum_{s_2\in S}  m(s_1,s_2,s_3)$.
Now, for sp-bisimulations $\Bi$ and $\Bi'$, also $\Bi;\Bi':=\{h \mathbin\in k;l \mid
k\mathbin\in \Bi \wedge l \mathbin\in \Bi'\}$
is an sp-bisimulation, showing transitivity of $\bis{sp}$.
\qed\end{proof}

\section{Strong bisimilarity}\label{sec:strong}

As discussed in the introduction and at the end of \Sect{operational},
Olderog defined a relation of \emph{strong bisimilarity} on safe Petri nets.

\begin{definition}\rm
For ${\cal B} \subseteq S_1 \times S_2$ a binary relation between the places
of two safe nets $N_i=(S_i,T_i,F_i,M_i,A_i,\ell_i)$, write $\widehat {\cal B}$ for
the set of all linkings $l\subseteq {\cal B}$ such that
$\pi_i(l)$ is a reachable marking of $N_i$ for $i\mathbin=1,2$ and
${\cal B} \cap \big(\pi_1(l)\times\pi_2(l)\big)=l$.
Now a \emph{strong bisimulation} as defined in \cite{Old91} can be
seen as a structure preserving bisimulation of the form \plat{$\widehat{\cal B}$}.
The nets $N_1$ and $N_2$ are \emph{strongly bisimilar} if $A_1\mathbin=A_2$ and
there is a linking $l$ in a strong bisimulation with $M_1\mathbin=\pi_1(l)$ and
$M_2 \mathbin=\pi_2(l)$.
\end{definition}
This reformulation of the definition from \cite{Old91} makes
immediately clear that strong bisimilarity of two safe Petri nets
implies their structure preserving bisimilarity.
Consequently, the \hyperlink{concurrency_requirement}{concurrency requirement}
for the net semantics from Olderog, as formalised by Requirements (1)
and (2) in \Sect{operational}, holds for structure preserving bisimilarity.

\section{Compositionality}\label{sec:congruence}
\newcommand{\ola}[1]{#1^l}
\newcommand{\ora}[1]{#1^r}

In this section I show that structure preserving bisimilarity is a
congruence for the operators of CCSP, or, in other words, that these
operators are compositional up to $\bis{sp}$.

\begin{theorem}\rm\label{thm:congruence}
If \mbox{$N_1 \mathbin{\bis{sp}} N_2$},
$a\mathbin\in Act$ and $R\subseteq Act \times Act$,
then $a_{\cal N}N_1 \bis{sp} a_{\cal N}N_2$ and
$R_{\cal N}(N_2) \mathbin{\bis{sp}} R_{\cal N}(N_2)$.
If $\ola{N_1} \mathbin{\bis{sp}} \ola{N_2}$
and $\ora{N_1} \mathbin{\bis{sp}} \ora{N_2}$
then $\ola{N_1}\|_{\cal N}\ora{N_1} \mathbin{\bis{sp}} \ola{N_2}\|_{\cal N}\ora{N_2}$ and,
if the initial markings of $N_i^l$ and $N_i^r$ are nonempty sets,
$\ola{N_1}+_{\cal N}\ora{N_1} \mathbin{\bis{sp}} \ola{N_2}+_{\cal N}\ora{N_2}$.
\end{theorem}

\begin{proof}
Let $N_i=(S_i,T_i,F_i,M_i,A_i,\ell_i)$ for $i\mathbin=1,2$,
and let $s_i$ and $u_i$ be the fresh place and transition introduced in
the definition of $a_{\cal N}N_i$.
From $N_1 \mathbin{\bis{sp}} N_2$ it follows that $A_1=A_2$ and hence $A_1\cup\{a\}=A_2\cup\{a\}$.

Let $\Bi$ be an sp-bisimulation containing a linking $k$ with
$M_i=\pi_i(k)$ for $i\mathbin=1,2$. Let $\Bi_a:= \Bi \cup \{h\}$,
with $h=\{(s_1,s_2)\}$. Then $h$ links the initial markings of
$a_{\cal N}N_1$ and $a_{\cal N}N_2$.
Hence it suffices to show that $\Bi_a$ is an sp-bisimulation.
So suppose $c\mathbin\leq h$ and $\pi_1(c)\mathbin=\precond{t_1}$ for some $t_1\mathbin\in T_1$.
Then $c\mathbin=h$ and $t_1 \mathbin= u_1$.
Take $t_2:= u_2$ and $\bar h:= \bar c := k$.

To show that $R_{\cal N}(N_2) \bis{sp} R_{\cal N}(N_2)$ it suffices to show that $\Bi$ also is an
sp-bisimulation between $R_{\cal N}(N_2)$ and $R_{\cal N}(N_2)$, which is straightforward.

Now let $\ola{N_i}\mathbin=(\ola{S_i},\ola{T_i},\ola{F_i},\ola{M_i},\ola{A_i},\ola{\ell_i})$
and $\ora{N_i}\mathbin=(\ora{S_i},\ora{T_i},\ora{F_i},\ora{M_i},\ora{A_i},\ora{\ell_i})$
for $i\mathbin=1,2$.\linebreak
Let $A:=\ola{A_1}\cap \ora{A_1}=\ola{A_2}\cap \ora{A_2}$.
Create the disjoint union of $\ola{N_i}$ and $\ora{N_i}$ in the definition of
$\ola{N_i}\|_{\cal N}\ora{N_i}$ by renaming all places $s$ and transitions $t$ of $\ola{N_i}$
into $s\|_A$ and $t\|_A$, and all places $s$ and transitions $t$ of $\ora{N_i}$
into ${_A}\|s$ and ${_A}\|t$.
Let $\ola{\Bi}$ and $\ora{\Bi}$ be sp-bisimulations containing linkings $\ola{k}$ and $\ora{k}\!$, respectively, with
$\ola{M_i}\mathop=\pi_i(\ola{k})$ and $\ora{M_i}\mathop=\pi_i(\ora{k})$, for $i\mathord=1,2$.
Take $\Bi:=\{(\ola{h} \|_A) + (_A\| \ora{h}) \mid \ola{h}\mathbin\in\ola{\Bi} \wedge \ora{h}\mathbin\in\ora{\Bi}\}$,
where $\ola{h}\|_A \mathbin{:=} \{(s_1\|_A,s_2\|_A)\linebreak[1] \mid (s_1,s_2)\mathbin\in \ola{h}\}$,
and $_A\| \ora{h}$ is defined likewise.
Then $\pi_i((\ola{k} \|_A) + (_A\| \ora{k})) = \pi_i(\ola{k})\|_A + {_A}\|\pi_i(\ora{k}) =
\ola{M_i}\|_A + {_A\|}\ora{M_i}$ is the initial marking of $\ola{N_i}\|_{\cal N}\ora{N_i}$ for $i\mathbin=1,2$, so
it suffices to show that $\Bi$ is an sp-bisimulation.

So suppose $c\mathbin\leq (h^l \|_A) + (_A\| h^r) \mathbin\in \Bi$
with $h^l\mathbin\in\Bi^l \wedge h^r\mathbin\in\Bi^r$ and
$\pi_1(c)\mathbin=\precond{t_1}$ for $t_1$ a transition of $N^l_1\|_{\cal N}N^r_1$.
Then $c$ has the form $(c^l \|_A) + (_A\| c^r)$ for $c^l\mathbin\leq h^l \mathbin\in \Bi^l$ and
$c^r\mathbin\leq h^r \mathbin\in \Bi^r\!$, and $t_1$ has the form (i) $t^l_1\|_A$ for $t^l_1\in T^l_1$
with $\ell^l_1(t^l_1)\mathbin{\notin} A$, or (ii) $(t^l_1\|_A,{_A}\|t^r_1)$ for $t^l_1\in T^l_1$ and $t^r_1\in T^r_1$
with $\ell^l_1(t^l_1)=\ell^r_1(t^r_1)\mathbin\in A$, or (iii) $_A\|t^r_1$ for $t^r_1\in T^r_1$
with $\ell^r_1(t^r_1)\mathbin{\notin} A$.
In case (i) one has $c^r\mathbin=\emptyset$ and $\pi_1(c^l)\mathbin=\precond{t^l_1}$, 
whereas in case (ii) $\pi_1(c^l)\mathbin=\precond{t^l_1}$ and $\pi_1(c^r)\mathbin=\precond{t^r_1}$.
I only elaborate case (ii); the other two proceed likewise.
Since $\Bi^l$ is an sp-bisimulation,
there are a transition $t_2^l$ with $\ell^l_2(t_2^l)=\ell^l_1(t^l_1)$ and $\pi_2(c^l)\mathbin=\precond{t^l_2}$,
and a linking $\bar c^l$ such that $\pi_1(\bar c^l)\mathbin=\postcond{t^l_1}$, $\pi_2(\bar c^l)\mathbin=\postcond{t^l_2}$
and $\bar h^l:=h^l-c^l+\bar c^l\in\Bi^l$.
Likewise, since $\Bi^r$ is an sp-bisimulation,
there are a transition $t_2^r$ with $\ell^r_2(t_2^r)=\ell^r_1(t^r_1)$ and $\pi_2(c^r)\mathbin=\precond{t^r_2}$,
and a linking $\bar c^r$ such that $\pi_1(\bar c^r)\mathbin=\postcond{t^r_1}$, $\pi_2(\bar c^r)\mathbin=\postcond{t^r_2}$
and $\bar h^r:=h^r-c^r+\bar c^r\mathbin\in\Bi^r$.
Take $t_2:=(t^l_2\|_A,{_A}\|t^r_2)$. This transition has the same label as $t^l_2$, $t^r_2$,
$t^l_1$, $t^r_1$ and $(t^l_1\|_A,{_A}\|t^r_1)=t_1$. Moreover,
$\pi_2(c)=\pi_2(c^l)\|_A + {_A}\|\pi_2(c^r) = \precond{t^l_2}\|_A + {_A}\|\precond{t^r_2} = \precond{t_2}$.
Take $\bar c := (\bar c^l \|_A) + (_A\| \bar c^r)$.
Then $\pi_1(\bar c)=\postcond{t_1}$, $\pi_2(\bar c)=\postcond{t_2}$
and $\bar h:= (h^l \|_A) + (_A\| h^r) -c+\bar c=  (\bar h^l \|_A) + (_A\| \bar h^r)\in \Bi$.

Let $\ola{N_i}\mathbin=(\ola{S_i},\ola{T_i},\ola{F_i},\ola{M_i},\ola{A_i},\ola{\ell_i})$
and $\ora{N_i}\mathbin=(\ora{S_i},\ora{T_i},\ora{F_i},\ora{M_i},\ora{A_i},\ora{\ell_i})$
for $i\mathbin=1,2$, with $M^l_i$ and $M^r_i$ nonempty plain sets,
but this time I assume the nets to already be disjoint, and such that all the
places and transitions added in the construction of $N_i^l+_{\cal N}N_i^r$ are fresh.
Let $\ola{\Bi}$ and $\ora{\Bi}$ be as above.
Without loss of generality I may assume that the linkings $h$ in $\ola{\Bi}$ and $\ora{\Bi}$
have the property that $\pi_i(h)$ is a reachable marking for $i\mathbin=1,2$, so that
the restriction of $\pi_i(h)$ to $\ola{M_i}$ or $\ora{M_i}$ is a plain set. Define
\[\Bi^+\mathbin{:=}\begin{array}[t]{l}
  \{h^l_\bullet + (h_+^l \otimes k^r) \mathbin{|} h_\bullet^l + h_+^l \mathbin\in\Bi^l \wedge
  h_+^l \lvertneqq k^l\} \\
  \{h_\bullet^r + (k^l \otimes h_+^r) \mathbin{|} h_\bullet^r + h_+^r \mathbin\in\Bi^r \wedge
  h_+^r \lvertneqq k^r\}
  \cup \{k^l\otimes k^r\}
\end{array}\]
where $h^l\otimes h^r := \{((s^l_1,s^r_1), (s^l_2,s^r_2)) \mid
(s^l_1,s^l_2)\mathbin\in h^l \wedge (s^r_1,s^r_2)\mathbin\in h^r\}$.
Now $\pi_i(k^l\otimes k^r)=\pi_i(k^l)\times \pi_i(k^r)=M_i^l\times M_i^r$ is the initial marking of $N^l_i+_{\cal N}N^r_i$, so
again it suffices to show that $\Bi^+$ is an sp-bisimulation.

So suppose $c \leq h^l_\bullet + (h_+^l \otimes k^r) \mathbin\in\Bi^+$
with $h_\bullet^l + h_+^l \mathbin\in\Bi^l$, $h_+^l \lvertneqq k^l$
and $\pi_1(c)\mathbin=\precond{t_1}$ for $t_1$ a transition of $N^l_1+_{\cal N}N^r_1$.

First consider the case that $c\leq h^l_\bullet$.
Then $c \leq  h^l_\bullet \leq h_\bullet^l + h_+^l \mathbin\in\Bi^l$.
Since $\Bi^l$ is an sp-bisimulation, there
  are a transition $t_2\mathbin\in T^l_2$ with $\ell^l_2(t_2)=\ell^l_1(t_1)$ and $\pi_2(c)=\precond{t_2}$,
  and a linking $\bar c$ such that $\pi_1(\bar c)=\postcond{t_1}$, $\pi_2(\bar c)=\postcond{t_2}$
  and $h_\bullet^l + h_+^r-c+ \bar c\in\Bi^l$.
Now $h^l_\bullet + (h_+^l \otimes k^r) - c + \bar c = (h^l_\bullet - c + \bar c) + (h^r_+ \otimes k_2) \mathbin\in\Bi^+$
because $(h_\bullet^l-c + \bar c) + h_+^r \mathbin\in\Bi^l$.

In the remaining case $\pi_1(c)$ contains a place $(s^l_1,s^r_1)\mathbin\in M^l_1\times M^r_1$, so
$t_1$ must have either the form
$t_{1l}^K$ with $\emptyset\neq K \leq \precond{t_1^l}\cap M^l_1$ for some $t_1^l\mathbin\in T_1^l$,
or $t_{1r}^K$ with $\emptyset\neq K \leq \precond{t_1^r}\cap M^r_1$ for some $t_1^r\mathbin\in T_1^r$.
First assume, towards a contradiction, that $t_1\mathbin=t_{1r}^K$. Then
$M_1^l {\times} K \mathbin\leq \precond{t_{1r}^K}\mathbin=\pi_1(c)\mathbin\leq
\pi_1(h_\bullet^l) + \pi_1(h^l_+\otimes k^r)$. Since the places in $M_1^l {\times} K \mathbin\subseteq
M_1^l{\times} M_1^r$ are fresh, it follows that
$M_1^l \times K \leq \pi_1(h^l_+\otimes k^r) \leq \pi_1(h^l_+)\times \pi_1(k^r) \leq \pi_1(h^l_+)\times M_1^r$,
implying that $M_1^l \leq \pi_1(h^l_+)$ and $K \leq  M_1^r$---here I use that $M_1^l\mathbin{\neq} \emptyset \mathbin{\neq} K$
and $\pi_1(h^l_+)$ and $M_1^r$ are plain sets. However, the condition
$h_+^l \lvertneqq k^l$ implies that $\pi_1(h_+^l) \mathbin{\lvertneqq} \pi_1(k^l)\mathbin=M_1^l$, yielding a contradiction.
Hence $t_1$ is of the form $t_{1l}^K$.

Since \plat{$\pi_1(c)\mathbin=\precond{t_{1l}^K}\mathbin=\precond{t_1^l}{-}K+(K\times M^r_1)$},
the linking $c$ must have the form $c_\bullet {+} c'$ with \plat{$\pi_1(c_\bullet)=\precond{t_1^l}-K$} and
$\pi_1(c')=K\times M^r_1$.
As no place in $\precond{t_1^l}-K$ can be in $M^l_1\times M^r_1 \supseteq \pi_1(h_+^l \otimes k^r)$,
it follows that $c_\bullet \leq h^l_\bullet$. Likewise, as none of the places in $K\times M^r_1$ can
be in $\pi_1(h^l_\bullet)$, it follows that $c' \leq h_+^l \otimes k^r$.
Thus $K \times M_1^r =\pi_1(c') \leq \pi_1(h^l_+\otimes k^r) \leq \pi_1(h^l_+)\times \pi_1(k^r) \leq \pi_1(h^l_+)\times M_1^r$,
implying $K\leq \pi_1(h^l_+)$---again using that $\pi_1(h^l_+)$ and $M_1^r\mathbin{\neq}\emptyset$ are plain sets.
The linking $h^l_+\otimes k^r$ has the property that its projection $\pi_1(h^l_+\otimes k^r)$ is a plain set.
Since a subset $c''$ of a such linking is completely determined by its first projection $\pi_1(c'')$,
it follows that $c'=c_+\otimes k^r$ for the unique linking $c_+\leq h^l_+$ with $\pi_1(c_+)=K$.

Now $c_\bullet + c_+ \leq h^l_\bullet + h^l_+ \mathbin\in\Bi^l$ and \vspace{-1pt}
$\pi_1(c_\bullet + c_+)\mathbin=(\precond{t_1^l}{-}K)+K\mathbin=\precond{t^l_1}$.
Since $\Bi^l$ is an sp-bisimulation, there
  are a transition $t^l_2\mathbin\in T^l_2$ with $\ell^l_2(t^l_2)\mathbin=\ell^l_1(t^l_1)$ and
  $\pi_2(c_\bullet{+} c_+)\mathbin=\precond{t^l_2}$,
  and a linking $\bar c$ such that $\pi_1(\bar c)=\postcond{t^l_1}$, $\pi_2(\bar c)=\postcond{t^l_2}$
  and $h_\bullet^l + h_+^l-(c_\bullet+ c_+) + \bar c\in\Bi^l$.
Let $L\mathbin{:=}\pi_2(c_+)$. Then $L\mathbin{\neq}\emptyset$ since $K\mathbin{\neq}\emptyset$,
$L=\pi_2(c_+)\leq \pi_2(h^l_+) \leq \pi_2(k^l) = M^l_2$ and \plat{$L=\pi_2(c_+) \leq \pi_2(c_\bullet+ c_+)=\precond{t_2^l}$}.
By \Def{operators} $N^l_2+_{\cal N}N_2^r$ has a transition $t^L_{2l}$ with
$\ell(t^L_{2l})\mathbin=\ell^l_2(t^l_2)\mathbin=\ell^l_1(t^l_1)\mathbin=\ell(t^L_{1l})$,
$\precond{t^L_{2l}}=\precond{t^l_2}{-}L+(L\times M^l_2)=\pi_2(c_\bullet+ c_+)-\pi_2(c_+)+(\pi_2(c_+)\times\pi_1(k^r))
=\pi_2(c_\bullet + (c_+\otimes k^r)) = \pi_2(c)$
and $\postcond{t^L_{2l}}=\postcond{t^l_2}=\pi_2(\bar c)$.
Moreover, $\pi_1(\bar c)\mathbin=\postcond{t^l_1}\mathbin=\postcond{t^K_1}\!$.
Finally, $h^l_\bullet + (h_+^l \otimes k^r) - c + \bar c = (h^l_\bullet - c_\bullet + \bar{c})
+ ((h_+^l\mathord-c_+)\otimes k^r)\in\Bi^+$
since $(h^l_\bullet - c_\bullet + c') + (h^l_+\mathord-c_+)\in\Bi^l$
and $h^l_+{-}c_+ \leq h^l_+ \lvertneqq k^l$.

The case supposing $c \leq h^r_\bullet + (k^r \otimes h_+^r) \mathbin\in\Bi^+$ follows by symmetry,
whereas the case $c\leq k^l\otimes k^r$ proceeds by simplification of the other two cases.
\qed
\end{proof}
\vfill

\section{Processes of nets and causal equivalence}\label{sec:processes}

A \emph{process} of a net $N$ \cite{Pe77,genrich80dictionary,GR83}
is essentially a conflict-free, acyclic net together
with a mapping function to $N$. It can be obtained by unwinding $N$,
choosing one of the alternatives in case
of conflict. It models a run, or concurrent computation, of $N$.
The acyclic nature of the process gives rise to a notion of causality
for transition firings in the original net via the mapping function.
A conflict present in the original net is represented by the existence of
multiple processes, each representing one possible way to decide the conflict.
This notion of process differs from the one used in process
algebra; there a ``process'' refers to the entire behaviour of a system,
including all its choices.

\begin{definition}\rm\label{df:process}
A \emph{causal net}%
\footnote{A causal net \cite{Pe77,Re13} is traditionally called an \emph{occurrence net} \cite{genrich80dictionary,GR83,Re85}.
  Here, following \cite{Old91}, I will
  not use the terminology ``occurrence net'' in order to avoid confusion
  with the occurrence nets of \cite{NPW81,Wi87a}; the latter extend
  causal nets with forward branching places, thereby capturing all runs
  of the represented system, together with the branching structure
  between them.}
is a net $\NN = (\SS, \TT, \FF, \MM_0, \AA, \ell_\NNs)$ satisfying
  \begin{itemize}
  \vspace{-1ex}
    \item $\forall s \in \SS. |\precond{s}| \leq\! 1\! \geq |\postcond{s}|
    \wedge\, \MM_0(s) = \left\{\begin{array}{@{}l@{\quad}l@{}}1&\mbox{if $\precond{s}=\emptyset$}\\
                                   0&\mbox{otherwise,}\end{array}\right.$
    \item $\FF$ is acyclic, i.e.,
      $\forall x \mathbin\in \SS \cup \TT. (x, x) \mathbin{\not\in} \FF^+$,
      where $\FF^+$ is the transitive closure of $\{(x,y)\mid \FF(x,y)>0\}$,
    \item and $\{t \in \TT \mid (t,u)\in \FF^+\}$ is finite for all $u\in \TT$.
  \vspace{-1ex}
  \end{itemize}
A \emph{folding} from a net $\NN = (\SS, \TT, \FF, \MM_0, \AA, \ell_\NNs)$ into a net $N = (S, T, F, M_0, A, \ell)$ is a function
$\rho:\SS \cup \TT \rightarrow S \cup T$ with $\rho(\SS) \subseteq S$ and $\rho(\TT) \subseteq T$, satisfying
  \begin{itemize}
  \vspace{-1ex}
    \item $\AA = A$ and $\ell_\NNs(t)=\ell(\rho(t))$ for all $t\in\TT$,
    \item $\rho(\MM_0) = M_0$, i.e.\ $M_0(s) = |\rho^{-1}(s) \cap \MM_0|$ for all $s\in S$, and
    \item $\forall t \in \TT, s \in S.\,
      F(s, \rho(t)) = |\rho^{-1}(s) \cap \precond{t}| \wedge
      F(\rho(t), s) = |\rho^{-1}(s) \cap \postcond{t}|$.
\footnote{For $H\subseteq \SS$, the multiset $\rho(H)\mathbin\in\IN^S$ is defined by $\rho(H)(s)=|\rho^{-1}(s)\cap H|$.
Using this, these conditions can be reformulated as $\rho(\precond{t})=\precond\rho(t)$ and
$\rho(\postcond{t})=\postcond{\rho(t)}$.}
  \vspace{-1ex} 
  \end{itemize}
A pair $\PP = (\NN, \rho)$ of a causal net $\NN$ and a folding of $\NN$ into a net $N$ is a \emph{process} of $N$.
$\PP$ is called \emph{finite} if $\TT$ is finite.
\end{definition}
Note that if $N$ has bounded parallelism, than so do all of its processes.

\begin{definition}\rm \cite{Old91}
A net $\NN$ is called a causal net \emph{of} a net $N$ if it is the
first component of a process $(\NN,\rho)$ of $N$.
Two nets $N_1$ and $N_2$ are \emph{causal equivalent}, notation $\equiv_{\it caus}$, if
they have the same causal nets.
\end{definition}
Olderog shows that his relation of strong bisimilarity is included in $\equiv_{\it caus}$ \cite{Old91},
and thereby establishes the
\hyperlink{concurrency_requirement}{concurrency requirement (1)} from
\Sect{operational} for $\equiv_{\it caus}$.

For $\NN= (\SS, \TT, \FF, \MM_0, \AA, \ell_\NNs)$ a causal net,
let $\NN^\circ := \{s \mathbin\in \SS \mid \postcond{s}\mathbin=\emptyset\}$.
The following result supports the claim that finite processes model finite runs.

\begin{proposition}\rm \cite[Theorems 3.5 and 3.6]{GR83}\label{prop:finite process marking}
$M$ is a reachable marking of a net $N$
iff $N$ has a finite process $(\NN,\rho)$ with $\rho(\NN^\circ)=M$.
Here $\rho(\NN^\circ)(s)=|\rho^{-1}(s) \cap \NN^\circ|$.
\end{proposition}

\noindent
A process is not required to represent a completed run of the original net.
It might just as well stop early. In those cases, some set of transitions can
be added to the process such that another (larger) process is obtained. This
corresponds to the system taking some more steps and gives rise to a natural
order between processes.

\begin{definition}\rm
  Let $\PP \mathbin= ((\SS, \TT, \FF, \MM_0, \AA, \ell), \rho)$ and $\PP' \mathbin=
  ((\SS', \TT\,', \FF\,', \MM_0', \AA', \ell'), \rho')$ be two processes of the same net.
\label{df:extension}
  $\PP'$ is a \emph{prefix} of $\PP$, notation $\PP'\!\leq \PP$, and 
  $\PP$ an \emph{extension} of $\PP'$, iff 
    $\SS'\subseteq \SS$,
    $\TT\,'\subseteq \TT$,
    $\MM_0' = \MM_0$,
    $\FF\,'=\FF\restrictedto(\SS' \mathord\times \TT\,' \mathrel\cup \TT\,' \mathord\times \SS')$
    and $\rho'=\rho\restrictedto(\SS'\cup \TT\,')$. (This implies that
    $\AA'=\AA$ and $\ell'=\ell\restrictedto \TT$.)
\end{definition}

\noindent
The requirements above imply that if $\PP'\!\leq \PP$, \plat{$(x,y)\mathbin\in
\FF^+$} and $y\mathbin\in \SS' \cup \TT\,'$ then $x\mathbin\in \SS' \cup \TT\,'\!$.
Conversely, any subset $\TT\,'\subseteq \TT$ satisfying
\plat{$(t,u)\in \FF^+ \wedge u\in \TT\,' \Rightarrow t\in \TT\,'$} uniquely determines a
prefix of $\PP$.
A process $(\NN,\rho)$ of a net $N$ is \emph{initial} if $\NN$ contains no transitions;
then $\rho(\NN^\circ)$ is the initial marking of $N$.
Any process has an initial prefix. 

\begin{proposition}\rm\cite[Theorem~3.17]{GR83}\label{prop:limit1}
If $\PP_i = ((\SS_i, \TT_i, \FF_i, {\MM_0}_i, \AA_i, \ell_i), \rho_i)$
$(i\mathbin\in\IN)$ is a chain of processes of a net $N$, satisfying $\PP_i \leq
\PP_j$ for $i\leq j$, then there exists a unique process
$\PP = ((\SS, \TT, \FF, {\MM_0}, \AA, \ell), \rho)$ of $N$ with $\SS = \bigcup_{i\in\IN}\SS_i$
and $\TT = \bigcup_{i\in\IN}\TT_i$---the \emph{limit} of this
chain---such that $\PP_i \leq \PP$ for all $i\mathbin\in\IN$.
\qed
\end{proposition}
In \cite{Pe77,genrich80dictionary,GR83} processes were
defined without the third requirement of \Def{process}.
Goltz and Reisig \cite{GR83} observed that certain
processes did not correspond with runs of systems, and proposed to
restrict the notion of a process to those that can be obtained as the
limit of a chain of finite processes \cite[end of Section~3]{GR83}.
By \cite[Theorems~3.18 and 2.14]{GR83}, for processes of finite nets
this limitation is equivalent with imposing the third bullet point of \Def{process}. 
My restriction to nets with bounded parallelism serves to recreate
this result for processes of infinite nets.

\begin{proposition}\rm\label{prop:limit2}
Any process of a net can be obtained as the limit of a chain of finite approximations.
\end{proposition}

\begin{proof}
Define the \emph{depth} of a transition $u$ in a causal net as
one more than the maximum of the depth of all transitions $t$ with $t F^+ u$.
Since the set of such transitions $t$ is finite, the depth of a
transition $u$ is a finite integer. Now, given a process $\PP$, the
approximation $\PP_i$ is obtained by restricting to those transitions
in $\PP$ of depth $\leq i$, together with all their pre- and
postplaces, and keeping the initial marking.
Clearly, these approximations form a chain, with limit $\PP$.
By induction on $i$ one shows that $\PP_i$ is finite.
For $\PP_0$ this is trivial, as it has no transitions.
Now assume $\PP_i$ is finite but $\PP_{i+1}$ is not.
Executing, in $\PP_{i+1}$, all transitions of $\PP_i$ one by one leads
to a marking of $\PP_{i+1}$ in which all remaining transitions of
$\PP_{i+1}$ are enabled. As these transitions cannot have common
preplaces, this violates the assumption that $\PP_{i+1}$ has bounded
parallelism.
\qed
\end{proof}

\section{A process-based characterisation of sp-bisimilarity}\label{sec:process-based}

This section presents an alternative characterisation of
sp-bisimilarity that will be instrumental in obtaining Theorems~\ref{thm:sp-caus}
and~\ref{thm:sp-fcb}, saying that $\bis{sp}$ is a finer semantic
equivalence than $\equiv_{\it caus}$ and $\approx_h$.
This characterisation could have been presented as the original
definition; however, the latter is instrumental in showing that
$\bis{sp}$ is coarser than $\approx_{pb}$ and $\equiv_{\it occ}$, and
implied by Olderog's strong bisimilarity.

\begin{definition}\rm
A \emph{process-based sp-bisimulation} between two nets $N_1$ and
$N_2$ is a set $\Ri$ of triples $(\rho_1,\NN,\rho_2)$ with $(\NN,\rho_i)$
a finite process of $N_i$, for $i\mathbin = 1,2$, such that
\vspace{-1ex}
\begin{itemize}
\item $\Ri$ contains a triple $(\rho_1,\NN,\rho_2)$ with $\NN$ a causal net containing no transitions,
\item if $(\rho_1,\NN,\rho_2) \mathbin\in \Ri$ and $(\NN',\rho'_i)$ with $i\mathbin\in\{1,2\}$
  is a fin.\ proc.\ of $N_i$ extending $(\NN,\rho_i)$
  then $N_j$ with $j\mathbin{:=}3\mathord-i$ has a process $(\NN',\rho'_j) \geq (\NN,\rho_j)$
  such that $(\rho'_1,\NN',\rho'_2)\in \Ri$.
\end{itemize}
\end{definition}

\begin{theorem}\rm\label{thm:process-based sp}
Two nets are sp-bisimilar iff there exists a process-based sp-bisimulation between them.
\end{theorem}

\begin{proof}
Let $\Ri$ be a process-based sp-bisimulation between nets $N_1$ and $N_2$.
Define $\Bi := \{ \{(\rho_1(\sk),\rho_2(\sk)) \mid \sk \in \NN^\circ\} \mid (\rho_1,\NN,\rho_2) \in\Ri\}$.
Then $\Bi$ is an sp-bisimulation:
\begin{itemize}
\vspace{-1ex}
\item Let $c\leq l\mathbin\in \Bi$ and $\pi_1(c)=\precond{t_1}$ for $t_1\mathbin\in T_1$.
  Then $l= \{(\rho_1(\sk),\rho_2(\sk) \mid \sk \in \NN^\circ\}$ for some $(\rho_1,\NN,\rho_2) \in\Ri$.
  Extend $\NN$ to $\NN'$ by adding a fresh transition $\tk$ and fresh places $s_i$ for $
  s\mathbin\in S_1$ and $i\mathbin\in\IN$ with $F_1(t_1,s)>i$; let $\precond{\tk}=\{\sk\in\NN^\circ \mid
  \rho_1(\sk)\mathbin\in\precond{t_1}\}$ and $\postcond{\tk}=\{s_i\mid s\mathbin\in S_1 \wedge
  i\mathbin\in\IN \wedge F_1(t_1,s)>i\}$. Furthermore, extend $\rho_1$ to $\rho_1'$ by
  $\rho'_1(\tk)\mathbin{:=} t_1$ and $\rho'_1(s_i)\mathbin{:=}s$. Then
  $\precond\rho'_1(\tk)\mathbin=\precond{t_1}\mathbin=\rho'_1(\precond{\tk})$ and
  $\postcond{\rho'_1(\tk)}=\postcond{t_1}\mathbin=\rho'_1(\postcond{\tk})$,
  so $(\NN',\rho'_1)$ is a process of $N_1$, extending $(\NN,\rho_1)$.
  Since $\Ri$ is a process-based sp-bisimulation, $N_2$ has a process $(\NN',\rho'_2) \geq (\NN,\rho_2)$
  such that $(\rho'_1,\NN',\rho'_2)\in \Ri$. Take $t_2:=\rho'_2(\tk)$.
  Then $\ell_2(t_2)=\ell_\NNs(\tk)=\ell_1(t_1)$ and $c= \{(\rho_1(\sk),\rho_2(\sk) \mid \sk \in \precond{\tk}\}$,
  so $\pi_2(c)=\{\rho_2(\sk)\mid \sk \mathbin\in \precond{\tk}\} = \rho_2(\precond{\tk}) = \rho'_2(\precond{\tk}) =
  \precond{\rho'_2(\tk)} = \precond{t_2}$. Take $c':=\{(\rho'_1(\sk),\rho'_2(\sk))\mid \sk\in\postcond{\tk}\}$.
  Then $\pi_1(c')=\postcond{t_1}$, $\pi_2(c')=\postcond{t_2}$ and
  $l':=l-c+c' = \{(\rho'_1(\sk),\rho'_2(\sk)) \mid \sk \in \NN^\circ {-} \precond{\tk} {+} \postcond{\tk}\} 
  = \{(\rho'_1(\sk),\rho'_2(\sk)) \mid \sk \in {\NN'}^\circ\} \in\Bi$.
\item The other clause follows by symmetry.\
\vspace{-1ex}
\end{itemize}
Since $\Ri$ contains a triple $(\rho_1,\NN,\rho_2)$ with $\NN$ a causal net containing no transitions,
$\Bi$ contains a linking $l := \{(\rho_1(\sk),\rho_2(\sk)) \mid \sk \in \NN^\circ$ such that
$\pi_i(l) = \rho_i(\NN^\circ) = M_i$ for $i\mathbin=1,2$, where $M_i$ is the initial marking of $N_i$.
Since $(\NN,\rho_i)$ is a process of $N_i$, $N_i$ must have the the same type as $\NN$, for $i\mathbin=1,2$.
It follows that $N_1 \bis{sp} N_2$.

Now let $\Bi$ be an sp-bisimulation between nets $N_1$ and $N_2$.
Let $\Ri := \{(\rho_1,\NN,\rho_2) \mid (\NN,\rho_i) \mbox{~is a finite process of~}N_i~(i\mathbin=1,2)\mbox{~and~}
 \{(\rho_1(\sk),\rho_2(\sk)) \mid \sk \in \NN^\circ\} \in \Bi \}$.
Then $\Ri$ is a process-based sp-bisimulation.
\begin{itemize}
\vspace{-1ex}
\item $\Bi$ must contain a linking $l$ with $\pi_i(l) = M_i$ for $i\mathbin=1,2$,
  where $M_i$ is the initial marking of $N_i$; let $l=\{(s_1^k,s_2^k)\mid k\in K\}$.
  Let $\NN$ be a causal net with places $\sk^k$ for $k\mathbin\in K$ and no transitions,
  and define $\rho_i$ for $i\mathbin=1,2$ by $\rho_i(\sk^k)=s_i^k$ for $k\mathbin\in K$.
  Then $(\NN,\rho_i)$ is an initial process of $N_i$ ($i\mathbin=1,2$) and $(\rho_1,\NN,\rho_2)\in\Ri$.
\item Suppose $(\rho_1,\NN,\rho_2) \mathbin\in \Ri$ and $(\NN',\rho'_1)$
  is a finite process of $N_1$ extending $(\NN,\rho_1)$. (The case of a finite process of $N_2$
  extending $(\NN,\rho_1)$ will follow by symmetry.)
  Then $l:=\{(\rho_1(\sk),\rho_2(\sk)) \mid \sk \mathbin\in \NN^\circ\} \mathbin\in \Bi$.
  Without loss of generality, I may assume that $\NN'$ extends $\NN$ by just one transition, $\tk$.
  The definition of a causal net ensures that $\precond{\tk}\subseteq \NN^\circ$, and the definition
  of a process gives $\rho_1'(\precond{\tk}) \mathbin= \precond{t_1}$, where $t_1:=\rho_1'(\tk)$.
  Let $c:=\{(\rho_1(\sk),\rho_2(\sk)) \mid \sk \in \precond{\tk}\}$. Then $c\leq l$ and
  $\pi_1(c)\mathbin=\rho_1(\precond{\tk})\mathbin=\rho'_1(\precond{\tk})\mathbin=\precond{t_1}$.
  Since $\Bi$ is an sp-bisimulation, there
  are a transition $t_2$ with $\ell(t_2)\mathbin=\ell(t_1)$ and $\pi_2(c)\mathbin=\precond{t_2}$,
  and a linking $c'$ such that  $\pi_1(c')\mathbin=\postcond{t_1}$, $\pi_2(c')\mathbin=\postcond{t_2}$
  and $l'\mathbin{:=}l-c+c'\mathbin\in\Bi$.
  The definition of a process gives $\rho_1'(\postcond{\tk}) \mathbin= \postcond{t_1}$.
  This makes it possible to extend $\rho_2$ to $\rho'_2$ so that $\rho_2'(\tk):=t_2$,
  $\rho'_2(\postcond{\tk})=\postcond{t_2}$ and
  $c' = \{(\rho'_1(\sk),\rho'_2(\sk)) \mid \sk \in \postcond{\tk}\}$.
  Moreover, $\rho'_2(\precond{\tk})=\rho_2(\precond{\tk})=\pi_2(c)=\precond{t_2}$.\vspace{-2pt}
  Thus $(\NN',\rho_2')$ is a finite process of $N_2$ extending $(\NN,\rho_2)$.
  Furthermore, $\{(\rho'_1(\sk),\rho'_2(\sk)) \mid \sk \in {\NN'}^\circ\} =
  \{(\rho'_1(\sk),\rho'_2(\sk)) \mid \sk \in {\NN}^\circ-\precond{\tk}+\postcond{\tk}\} =
  l-c+c' \in \Bi$. Hence $(\rho'_1,\NN',\rho'_2)\in \Ri$.
\qed
\vspace{-1ex}
\end{itemize}
\end{proof}

\section{Relating sp-bisimilarity to other semantic equivalences}\label{sec:relating}

In this section I place sp-bisimilarity in the spectrum of existing
semantic equivalences for nets, as indicated in \Fig{spectrum}.

\subsection{Place bisimilarity}

The notion of a place bisimulation, defined in \cite{ABS91}, can be
reformulated as follows.

\begin{definition}\rm
A \emph{place bisimulation} is a structure preserving bisimulation of
the form \plat{$\overline{\cal B}$} (where \plat{$\overline{\cal B}$}
is defined in \Sect{sp-bis}).
Two nets $N_i=(S_i,T_i,F_i,M_i,A_i,\ell_i)$ ($i\mathbin=1,2$)
are \emph{strongly bisimilar}, notation $N_1 \approx_{pb} N_2$, if $A_1\mathbin=A_2$ and
there is a linking $l$ in a place bisimulation with $M_1\mathbin=\pi_1(l)$ and $M_2 \mathbin=\pi_2(l)$.
\end{definition}
It follows that $\approx_{pb}$ is finer than $\bis{sp}$, in the sense
that place bisimilarity of two nets implies their structure preserving bisimilarity.

\subsection{Occurrence net equivalence}\label{ssec:occ}

Definitions of the \emph{unfolding} for various classes of Petri nets into an \emph{occurrence net}
appear in \cite{NPW81,Wi84,Wi87a,GV87,En91,MMS97,vG05c}---I will not repeat them here. In all these cases, the
definition directly implies that if an occurrence net $\NN$ results from unfolding a net $N$
then $\NN$ is safe and there exists a folding of $\NN$ into $N$ (recall \Def{process}) satisfying
  \begin{itemize}
  \vspace{-1ex}
    \item if $\MM$ is a reachable marking of $\NN$, and $t\mathbin\in T$ is a transition of $N$ with
      $\precond{t} \leq \rho(\MM)$ then there is a $\tk\mathbin\in\TT$ with $\rho(\tk)=t$.
  \end{itemize}

\begin{proposition}\rm\label{prop:unfolding}
If such a folding from $\NN$ to $N$ exists, then $\NN\bis{sp}N$.
\end{proposition}
\begin{proof}
The set of linkings $\Bi:=\{\{(\sk,\rho(\sk)) \mid \sk\mathbin\in \MM\}\mid \MM \mbox{~a reachable
  marking of~}\NN\}$ is an sp-bisimulation between $\NN$ and $N$. Checking this is trivial.
\qed
\end{proof}
Two nets $N_1$ and $N_2$ are \emph{occurrence net equivalent} \cite{GV87} if they have isomorphic
unfoldings. Since isomorphic nets are strongly bisimilar \cite{Old91} and hence structure preserving
bisimilar, it follows that occurrence net equivalence between nets is finer than structure
preserving bisimilarity.

In \cite{ABS91} it is pointed out that the strong bisimilarity of Olderog ``is not compatible with
unfoldings'': they show two nets that have isomorphic unfoldings, yet are not strongly bisimilar.
However, when the net $N$ is safe, the sp-bisimulation displayed in the proof of \Prop{unfolding}
is in fact a strong bisimulation, showing that each net is strongly bisimilar with its unfolding.
This is compatible with the observation of \cite{ABS91} because of the non-transitivity of strong bisimilarity.
\vfill

\subsection{Causal equivalence}

Causal equivalence is coarser than structure preserving bisimilarity.

\begin{theorem}\rm\label{thm:sp-caus}
If $N_1\bis{sp}N_2$ for nets $N_1$ and $N_2$, then $N_1\equiv_{\it caus}N_2$.
\end{theorem}

\begin{proof}
By \Thm{process-based sp} there exists a process-based sp-bisimulation $\Ri$ between $N_1$ and $N_2$.
$\Ri$ must contain a triple $(\rho^0_1,\NN^0,\rho^0_2)$ with $\NN^0$ a causal net containing no transitions.
So \plat{$(\NN^0\!,\rho^0_1)$} and \plat{$(\NN^0\!,\rho^0_2)$} are initial processes of $N_1$ and $N_2$, respectively. 
The net \plat{$\NN^0$} contains isolated places only, as many as the size of the initial markings of $N_1$
and $N_2$.

Let $\NN$ be a causal net of $N_1$. I have to prove that $\NN$ is also a causal net of $N_2$.
Without loss of generality I may assume that $\NN^0$ is a prefix of $\NN$, as being a causal net of
a given Petri net is invariant under renaming of its places and transitions.

So $N_1$ has a process $\PP_1=(\NN,\rho_1)$.
By \Prop{limit2}, $\PP_1$ is the limit of a chain $\PP^0_1 \mathbin\leq \PP^1_1 \mathbin\leq \PP^2_1
\mathbin\leq \dots$ of finite processes of $N_1$. Moreover, for $\PP_1^0$ one can take $(\NN^0\!,\rho^0_1)$.
Let $\PP^i_1=(\NN^i\!,\rho_1^i)$ for $i\mathbin\in\IN$.
By induction on $i\mathbin\in\IN$, it now follows from the properties of a process-based sp-bisimulation that
$N_2$ has processes $\PP_2^{i+1}\mathbin=(\NN^{i+1}\!,\rho_2^{i+1})$, such that
$(\NN^{i}\!,\rho_2^{i})\mathbin\leq(\NN^{i+1}\!,\rho_2^{i+1})$
and $(\rho^{i+1}_1\!,\NN^{i+1}\!,\rho^{i+1}_2)\mathbin\in\Ri$.
Using \Prop{limit1}, the limit $\PP_2=(\NN,\rho_2)$ of this chain is a process of $N_2$,
contributing the causal net $\NN$.
\qed
\end{proof}

\subsection{History preserving bisimilarity}

The notion of \emph{history preserving bisimilarity} was originally proposed in \cite{RT88} under the
name \emph{behavior structure bisimilarity}, studied on event structures in \cite{GG01},
and first defined on Petri nets, under to the individual token interpretation, in \cite{BDKP91},
under the name \emph{fully concurrent bisimulation} equivalence.

\begin{definition}\rm\cite{BDKP91}
Let $\NN_i= (\SS_i, \TT_i, \FF_i, {\MM_0}_i, \AA_i, \ell_i)$ ($i\mathbin=1,2$) be two causal nets.
An \emph{order-isomorphism} between them is a bijection $\beta:\TT_1\rightarrow\TT_2$ such that
$\AA_1=\AA_2$,
$\ell_2(\beta(t))=\ell_1(t)$ for all $t\mathbin\in \TT_1$,
and
$t \mathrel{\FF_1^+} u$ iff $\beta(t) \mathrel{\FF_2^+} \beta(u)$ for all $t,u\in T_1$.
\end{definition}

\begin{definition}\rm\cite{BDKP91}
A \emph{fully concurrent bisimulation} between two nets $N_1$ and
$N_2$ is a set $\Ri$ of triples $((\rho_1,\NN_1),\beta,(\NN_2,\rho_2))$ with $(\NN_i,\rho_i)$
a finite process of $N_i$, for $i\mathbin = 1,2$, and $\beta$ an order-isomorphism between $\NN_1$
and $\NN_2$, such that
\vspace{-1ex}
\begin{itemize}
\item $\Ri$ contains a triple $((\rho_1,\NN_1),\beta,(\NN_2,\rho_2))$ with $\NN_1$ containing no transitions,
\item if $(\PP_1,\beta,\PP_2) \mathbin\in \Ri$ and $\PP'_i$ with $i\mathbin\in\{1,2\}$
  is a fin.\ proc.\ of $N_i$ extending $\PP_i$,
  then $N_j$ with $j\mathbin{:=}3\mathord-i$ has a process $\PP'_j \geq \PP_j$
  such that $(\PP'_1,\beta',\PP'_2)\in \Ri$ for some $\beta'\supseteq\beta$.%
\vspace{-1ex}
\end{itemize}
Write $N_1 \approx_{h} N_2$ or $N_1 \approx_{\it fcb} N_2$ iff such a bisimulation exists.
\end{definition}
It follows immediately from the process-based characterisation of sp-bisimilarity
in \Sect{process-based} that fully concurrent bisimilarity (or history preserving bisimilarity based
on the individual token interpretation of nets) is coarser than sp-bisimilarity.

\begin{theorem}\rm\label{thm:sp-fcb}
If $N_1\bis{sp}N_2$ for nets $N_1$ and $N_2$, then $N_1\approx_{\it fcb}N_2$.
\end{theorem}

\begin{proof}
A process-based sp-bisimulation is simply a fully concurrent bisimulation with the extra requirement
that $\beta$ must be the identity relation.
\qed
\end{proof}

\section{Inevitability for non-reactive systems}\label{sec:inevitability}

A run or execution of a system modelled as Petri net $N$ can be formalised as a path of $N$
(defined in \Sect{nets}) or a process of $N$ (defined in \Sect{processes}). A path or process
representing a complete run of the represented system---one that is not just
the first part of a larger run---is sometimes called a \emph{complete} path or process.
Once a formal definition of a complete path or process is agreed upon, an action $b$ is
\emph{inevitable} in a net $N$ iff each complete path (or each complete process) of $N$ contains a
transition labelled $b$. In case completeness is defined both for paths and processes, the
definitions ought to be such that they give rise to the same concept of inevitability.

The definition of which paths or processes count as being complete depends on two factors:
(1) whether actions that a net can perform by firing a transition are fully under control of the
represented system itself or (also) of the environment in which it will be running, and
(2) what type of progress or fairness assumption one postulates to guarantee that actions that are
due to occur will actually happen. In order to address (2) first, in this section I deal only with
nets in which all activity is fully under control of the represented system. In \Sect{reactive}
I will generalise the conclusions to reactive systems.

When making no progress or fairness assumptions, a system always has
the option not to progress further, and all paths and all processes are complete---in
particular initial paths and processes, containing no transitions. Consequently, no action is
inevitable in any net, so each semantic equivalence respects inevitability.

When assuming progress, but not justness or fairness, any infinite path or process is complete, and
a finite path or process is complete iff it is maximal, in the sense that it has no proper extension.
In this setting, interleaving bisimilarity, and hence also each finer equivalence, respects
inevitability. The argument is that an interleaving bisimulation induces a relation between the
paths of two related nets $N_1$ and $N_2$, such that
\begin{itemize}
\vspace{-1ex}
\item each path of $N_1$ is related to a path of $N_2$ and vice versa,
\item if two paths are related, either both or neither contain a transition labelled $b$,
\item if two paths are related, either both or neither of them are complete.
\vspace{-1ex}
\end{itemize}

In the rest of this paper I will assume justness, and hence also progress, but not (weak or strong)
fairness, as explained in \SSect{inevitability}. In this setting a process is \emph{just} or \emph{complete}%
\footnote{The term ``complete'' is meant to vary with the choice of a progress or fairness
  assumption; when assuming only justness, it is set to the value ``just''.}
iff it is maximal, in the sense that it has no proper extension.
\vspace{-2ex}

\noindent
\begin{wrapfigure}[4]{r}{0.385\textwidth}
\input{inevitable}
\centerline{\box\graph}
\end{wrapfigure}
\begin{example}
The net depicted on the right has a complete process performing the action $a$ infinitely often,
but never the action $b$. It consumes each token that is initially present or stems from any
firing of the transition $t^a$. Hence $b$ is not inevitable. This fits with the intuition that if a
transition occurrence is perpetually enabled it will eventually happen---but only when strictly
adhering to the individual token interpretation of nets.  Under this interpretation, each firing of
$t^b$ using a particular token is a different transition occurrence. It is possible to
schedule an infinite sequence of $a$s in such a way that none such transition occurrence is
perpetually enabled from some point onwards.

When adhering to the collective token interpretation of nets, the
action $b$ could be considered inevitable, as
in any execution scheduling $a$s only, transition $t^b$ is perpetually enabled.
Since my structure preserving bisimulation fits within the individual token interpretation,
here one either should adhere to that interpretation, or restrict attention to safe nets, where
there is no difference between both interpretations.
\end{example}

\section{History preserving bisimilarity does not respect inevitability}\label{sec:counterexample}

\begin{figure}
\vspace{-4ex}
\input{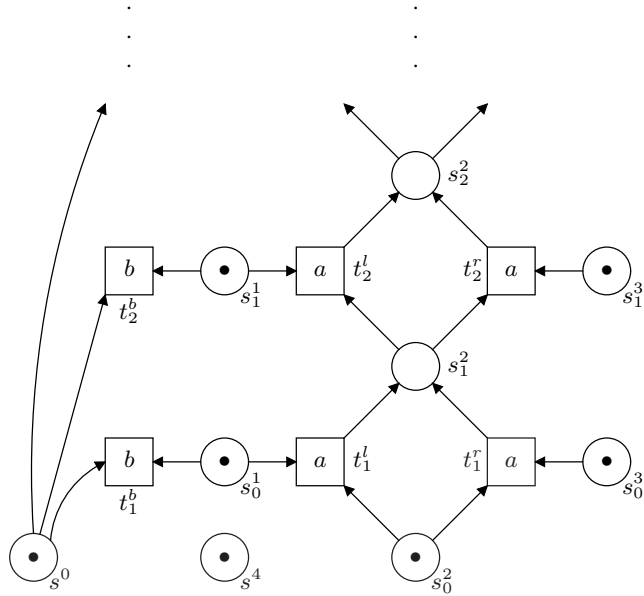}
\centerline{\box\graph}
\caption{A net in which the action $b$ is not inevitable}\label{fig:evitable}
\end{figure}

\noindent
Consider the safe net $N_1$ depicted in \Fig{evitable}, and the net $N_2$
obtained from $N_1$ by exchanging for any transition $t^b_i$
($i\mathord>0$) the preplace $s^1_{i-1}$ for $s^4$.
The net $N_2$ performs in parallel an infinite sequence of $a$-transitions
(where at each step $i{>}0$ there is a choice between $t^l_i$ and $t^r_i$)
and a single $b$-transition (where there is a choice between $t^b_i$ for $i{>}0$).
In $N_2$ the action $b$ is inevitable. In $N_1$, on the other hand, $b$ is not inevitable, for the
run of $N_1$ in which $t^l_i$ is chosen over $t^r_i$ for all $i{>}0$ is complete, and cannot be
extended which a $b$-transition. Thus, each semantic equivalence that equates $N_1$ and $N_2$ fails
to respect inevitability.

\begin{theorem}\rm
Causal equivalence does not respect inevitability.
\end{theorem}

\begin{proof}
$N_1 \equiv_{\it caus} N_2$, because both nets have the same causal nets.
One of these nets is depicted in \Fig{causal nets}; the others are obtained by omitting the
$b$-transition, and/or omitting all but a finite prefix of the $a$-transitions.
\qed
\end{proof}

\begin{figure}
\input{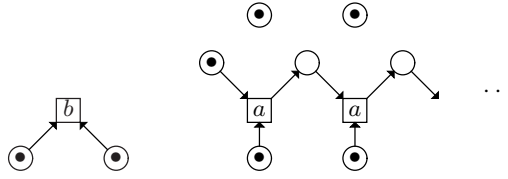}
\centerline{\box\graph}
\caption{A causal net of $N_1$ and $N_2$}\label{fig:causal nets}
\end{figure}

\begin{theorem}\rm
History preserving bisimilarity does not respect inevitability.
\end{theorem}

\begin{proof}
Recall that $N_1$ and $N_2$ differ only in their flow relations, and have the same set of transitions.
I need to describe a fully concurrent bisimulation $\Ri$ between $N_1$ and $N_2$.
$\Ri$ consists of a set of triples, each consisting of a process of $N_1$, a related process of
$N_2$, and an order isomorphism between them.
First of all I include all triples $(\PP_1,\beta,\PP_2)$ where $\PP_1$ is an arbitrary process of
$N_1$, $\PP_2$ is the unique process of $N_2$ that induces the same set of transitions as $\PP_1$,
and $\beta$ relates transition of $\PP_1$ and $\PP_2$ when they map to the same transition of $N_i$
($i{=}1,2$).
Secondly, I include all triples $(\PP_1,\beta,\PP_2)$ where $\PP_2$ is an arbitrary process of
$N_2$ inducing both $t^b_k$ and $t^l_k$ for some $k{>}0$, and $\PP_1$ is any process of $N_1$
that induces the same transitions as $\PP_2$ except that, for some $h{\geq}k$ the induced transition
$t^l_h$, if present, is replaced by $t^r_h$, and $t^b_k$ is replaced by $t^b_h$.
($\beta$ should be obvious.)
It is trivial to check that the resulting relation is a fully concurrent bisimulation indeed.
\qed
\end{proof}

\section{Structure preserving bisimilarity respects inevitability}\label{sec:respects}

\begin{definition}\rm
A net $\NN$ is called a \emph{complete} causal net of a net $N$ if it is the
first component of a maximal process $(\NN,\rho)$ of $N$.
Two nets $N_1$ and $N_2$ are \emph{complete causal net equivalent}, notation $\equiv_{\it cc}$, if
they have the same complete causal nets.
\end{definition}
Since the causal nets of a net $N$ are completely determined by the complete causal nets of $N$,
namely as their prefixes, $N_1\equiv_{\it cc}N_2$ implies $N_1\equiv_{\it caus.}N_2$.
It follows immediately from the definition of inevitability that $\equiv_{\it cc}$ respects inevitability.
Thus, to prove that $\bis{sp}$ respects inevitability it suffices to show that $\bis{sp}$ is finer
than $\equiv_{\it cc}$.

\begin{theorem}\rm
If $N_1\bis{sp}N_2$ for nets $N_1$ and $N_2$, then $N_1\equiv_{\it cc}N_2$.
\end{theorem}

\begin{proof}
Suppose $N_1\mathbin{\bis{sp}}N_2$.
By \Thm{process-based sp} there exists a process-based sp-bisimulation $\Ri$ between $N_1$ and $N_2$.
$\Ri$ must contain a triple \plat{$(\rho^0_1,\NN^0,\rho^0_2)$} with \plat{$\NN^0$} a causal net containing no transitions.
So \plat{$(\NN^0\!,\rho^0_1)$} and \plat{$(\NN^0\!,\rho^0_2)$} are initial processes of $N_1$ and $N_2$, respectively. 
The net \plat{$\NN^0$} contains isolated places only.

Let $\NN$ be a complete causal net of $N_1$. I have to prove that $\NN$ is also a complete causal net of $N_2$.
Without loss of generality I may assume that $\NN^0$ is a prefix of $\NN$, as being a complete
causal net of a given Petri net is invariant under renaming of its places.

So $N_1$ has a complete process $\PP_1=(\NN,\rho_1)$.
By \Prop{limit2}, $\PP_1$ is the limit of a chain $\PP^0_1 \mathbin\leq \PP^1_1 \mathbin\leq \PP^2_1
\mathbin\leq \dots$ of finite processes of $N_1$. Moreover, for $\PP_1^0$ one can take $(\NN^0\!,\rho^0_1)$.
Let $\PP^i_1=(\NN^i\!,\rho_1^i)$ for $i\mathbin\in\IN$.
By induction on $i\mathbin\in\IN$, it now follows from the properties of a process-based sp-bisimulation that
$N_2$ has processes $\PP_2^{i+1}\mathbin=(\NN^{i+1}\!,\rho_2^{i+1})$, such that
$(\NN^{i}\!,\rho_2^{i})\mathbin\leq(\NN^{i+1}\!,\rho_2^{i+1})$
and $(\rho^{i+1}_1\!,\NN^{i+1}\!,\rho^{i+1}_2)\mathbin\in\Ri$.
Using \Prop{limit1}, the limit $\PP_2=(\NN,\rho_2)$ of this chain is a process of $N_2$.
It remains to show that $\PP_2$ is complete.

Towards a contradiction, let $\PP_{2u}\mathbin=(\NN_{u},\rho_{2u})$ be a proper extension of $\PP_2$, say
with just one transition, $u$. Then $\precond{u} \subseteq \NN^\circ$.
By the third requirement on occurrence nets of \Def{process}, their are only finitely many
transitions $t$ with $(t,u)\mathbin\in \FF^+_{2u}$. Hence one of the finite approximations \plat{$\NN^k$} of
$\NN$ contains all these transitions. So \plat{$\precond{u} \subseteq (\NN^k)^\circ$}. Let, for all
$i\mathbin\geq k$, $\PP^i_{2u}\mathbin=(\NN_u^i,\rho^i_{2u})$ be the finite prefix of $\PP_{2u}$ that extends
$\PP^i_2$ with the single transition $u$. Then $\PP^i_{2u} \leq \PP^{i+1}_{2u}$ for all $i\mathbin\geq k$,
and the limit of the chain $\PP^k_{2u} \mathbin\leq \PP^{k+1}_{2u} \mathbin\leq \dots$ is $\PP_{2u}$.
By induction on $i\mathbin\in\IN$, it now follows from the properties of a process-based sp-bisimulation that
$N_1$ has processes \plat{$\PP_{1u}^{i}\mathbin=(\NN^{i}_u,\rho_{1u}^{i})$} for all $i\mathbin\geq k$, such that
$(\rho^{i}_{1u},\NN^{i}_u,\rho^{i}_{2u})\mathbin\in\Ri$,
$(\NN^{k}\!,\rho_1^{k})\mathbin\leq(\NN^{k}_u,\rho_{1u}^{k})$ and
$(\NN^{i}_u,\rho_{1u}^{i})\mathbin\leq(\NN_u^{i+1},\rho_{1u}^{i+1})$.
Using \Prop{limit1}, the limit $\PP_{1u}=(\NN_u,\rho_{1u})$ of this chain is a process of $N_1$.
It extends $\PP_1$ with the single transition $u$, contradicting the maximality of $\PP_1$.
\qed
\end{proof}

\section{Inevitability for reactive systems}\label{sec:reactive}

In the modelling of reactive systems, an action performed by a net is typically a synchronisation
between the net itself and its environment. Such an action can take place only when the net is ready
to perform it, as well as its environment. In this setting, an adequate formalisation of the concepts of
justness and inevitability requires keeping track of the set of actions that from some point
onwards are blocked by the environment---e.g.\ because the environment is not ready to partake in
the synchronisation. Such actions are not required to occur eventually, even when they are
perpetually enabled by the net itself. Let's speak of a \emph{$Y$-environment} if $Y$ is this set of
actions. In \Sect{inevitability} I restricted attention to $\emptyset$-environments, in which an
action can happen as soon as it is enabled by the net in question.
In \cite{GH15a} a path is called $Y$-just iff, when assuming justness, it models a
complete run of the represented system in a $Y$-environment. The below is a formalisation for this
concept for Petri nets under the individual token interpretation.

\begin{definition}\rm
A process of a net is \emph{$Y$-just} or \emph{$Y$-complete} it each of its proper extensions adds a
transition with a label in $Y$.
\end{definition}
Note that a just or complete process as defined in \Sect{inevitability} is a $\emptyset$-just or
$\emptyset$-complete process.
In applications there often is a subset of actions that are known to be fully controlled by the
system under consideration, and not by its environment. Such actions are often called \emph{non-blocking}.
A typical example from process algebra \cite{Mi90ccs} is the internal action $\tau$.
In such a setting, $Y$-environments exists only for sets of actions $Y\subseteq \HC$,
where $\HC$ is the set of all non-non-blocking actions.

A process of a net is \emph{complete} if it models a complete run of the represented
system in some environment. This is the case iff it is $Y$-complete for some set $Y\subseteq\HC$,
which is the case iff it is $\HC$-complete.

In \cite{Re13}, non-blocking is a property of transitions rather than actions, and non-blocking
transitions are called \emph{hot}. Transitions that are not hot are \emph{cold}, which inspired my
choice of the latter $\HC$ above.  In this setting, a process $\PP=(\NN,\rho)$ is complete iff
the marking $\rho(\NN^\circ)$ enables cold transitions only \cite{Re13}.

\begin{definition}\rm
A action $b$ is \emph{$Y$-inevitable} in a net if each $Y$-complete process contains a transition
labelled $b$.
A semantic equivalence $\approx$ \emph{respects $Y$-inevitability} if whenever $N_1\approx N_2$ and
$b$ is $Y$-inevitable in $N_1$, then $b$ is $Y$-inevitable in $N_2$.
It \emph{respects inevitability} iff it respects $Y$-inevitability for each $Y\subseteq\HC$.
\end{definition}

\noindent
In \Sect{counterexample} it is shown that $\equiv_{\it caus}$ and $\approx_h$ do not respect
$\emptyset$-inevitability. From this it follows that they do not respect inevitability.
In \Sect{respects} it is shown that $\bis{sp}$ does respect $\emptyset$-inevitability.
By means of a trivial adaptation the same proof shows that $\bis{sp}$ respects
$Y$-inevitability, for arbitrary $Y$. All that is needed is to assume that the transition $u$
in that proof has a label $\notin Y$. Thus $\bis{sp}$ respects inevitability.

\section{Conclusion}

This paper proposes a novel semantic equivalence for current systems
represented as Petri nets: \emph{structure preserving bisimilarity}.
As a major application---the one that inspired this work---it is used
to establish the agreement between the operational Petri net semantics
of the process algebra CCSP as proposed by Olderog, and its
denotational counterpart. An earlier semantic relation used for this
purpose was Olderog's \emph{strong bisimilarity} on safe Petri nets,
but that relation failed to be transitive. I hereby conjecture that on
the subclass of occurrence nets, strong bisimilarity and structure
preserving bisimilarity coincide. If this it true, it follows,
together with the observations of \Sect{strong} that strong
bisimilarity is included in structure preserving bisimilarity, and of
\SSect{occ} that each safe net is strongly bisimilar with its
unfolding into an occurrence net, that on safe nets structure
preserving bisimilarity is the transitive closure of strong bisimilarity.

\SSect{criteria} proposes nine requirements on a semantic equivalence
that is used for purposes like the one above. I have shown that
structure preserving bisimilarity meets eight of these requirements
and conjecture that it meets the remaining one as well.
\begin{list}{\normalfont\bfseries --}{\leftmargin 15pt}
\vspace{-1ex}
\item It meets Requirement 1, that it respects branching time, as a
  consequence of \Thm{sp-fcb}, saying that it is finer than history
  preserving bisimilarity, which is known to be finer than
  interleaving bisimilarity.
\item It meets Requirement 2, that it fully captures causality and
  concurrency (and their interplay with branching time),\footnote{when taking the
  individual token interpretation of nets, or restricting\label{fn}
  attention to safe ones} also as a consequence of \Thm{sp-fcb}.
\item It meets Requirement 3, that it respects inevitability (under
  the standard interpretation of Petri nets that assumes justness but
  not fairness),$^{\mbox{\scriptsize \ref{fn}}}$ as shown in \Sect{respects}.
\item It meets Requirement 4, that it is real-time consistent, as a
  result of \Thm{sp-fcb}.
\item I conjecture that it meets Requirement 5, that it is preserved
  under action refinement.
\item It meets Requirement 6, that it is finer than causal equivalence,
  by \Thm{sp-caus}.
\item It meets Requirement 7, that it is coarser than $\equiv_{\it occ}$,
  as shown in \SSect{occ}.
\item It meets Requirement 8, that it is a congruence for the
  CCSP operators, by Thm.~\ref{thm:congruence}.
\item It meets Requirement 9, that it allows to establish agreement between the
  operational and denotational interpretations of CCSP operators,
  since it is coarser than Olderog's strong bisimilarity, as shown in \Sect{strong}.
\pagebreak[2]
\end{list}
Moreover, structure preserving bisimilarity is the first known
equivalence that meets these requirements. In fact, it is the first
that meets the key Requirements 3, 4, 7 and 9.

\paragraph{Acknowledgement}
My thanks to Ursula Goltz for proofreading and valuable feedback.


\end{document}